%% file: main.tex
\documentclass[journal]{IEEEtran}
\usepackage[utf8]{inputenc}
\input{header.tex}
 \title{ Joint Spatial-Propagation Modeling \\
 of Cellular Networks Based on the\\ Directional Radii of Poisson Voronoi Cells}
 \author{\IEEEauthorblockN{Ke Feng,  \textit{University of Notre Dame}, and}
\and
\IEEEauthorblockN{Martin Haenggi, \textit{University of Notre Dame}}
\thanks{This work has been supported in part by the US National Science Foundation through Grant 2007498.}
}

\begin{document}
\maketitle

\begin{abstract}
In coverage-oriented networks, base stations (BSs) are deployed in a way such that users at the cell boundaries achieve sufficient signal strength. The shape and size of cells vary from BS to BS, since the large-scale signal propagation conditions differ in different geographical regions. This work proposes and studies a joint spatial-propagation (JSP) model, which considers the correlation between cell radii and the large-scale signal propagation (captured by shadowing).  

We first introduce the notion of the directional radius of Voronoi cells, which has applications in cellular networks and beyond. The directional radius of a cell is defined as the distance from the nucleus to the cell boundary at an angle relative to the direction of a uniformly random location in the cell. We study the distribution of the radii in two types of cells in the Poisson Voronoi tessellations: the zero-cell, which contains the origin, and the typical cell. 

The results are applied to analyze the JSP model. We show that, even though the Poisson point process (PPP) is often considered as a pessimistic spatial model for BS locations, the JSP model with the PPP achieves coverage performance close to the most optimistic one---the standard triangular lattice model. Further, we show that the network performance depends critically on the variance of the large-scale path loss along the cell boundary.

\end{abstract}


\input{intro.tex}
\input{Distances.tex}

\input{sys-model.tex}

\input{analysis_JSP.tex}

 \input{Conclusions}

\input{appendix.tex}

\bibliography{ref,ref1}
\bibliographystyle{IEEEtran}
\end{document}

%% file: header.tex
\usepackage{amsmath}
\usepackage{mathtools}
\usepackage{amsfonts}
\usepackage[font=footnotesize]{subcaption}

\usepackage[labelformat=simple]{subcaption}

\DeclareCaptionLabelSeparator{periodspace}{.\quad}
\usepackage{graphicx}
\usepackage{textcomp}
\usepackage{bbm}
\usepackage{amssymb}
\usepackage{hyperref}
\usepackage{url}
\usepackage{verbatim}

\usepackage{amsthm}
\usepackage{amssymb}
\usepackage{amsmath}
\usepackage{graphicx}
\usepackage{textcomp}
\usepackage{cite}

\def\Eii{\mathrm{E}_1}
\def\Eii{\operatorname{E_1}}

\def\E{\mathbb{E}}
\def\P{\mathbb{P}}
\def\R{\mathbb{R}}

\def\argmin{\operatorname{arg~min}}
\def\argmax{\operatorname{arg~max}}

\def\dd{{\rm d}}

\def\rot{\operatorname{rot}}

\theoremstyle{plain}

\newtheorem{theorem}{Theorem}
\newtheorem{fact}{Fact}
\newtheorem{remark}{Remark}
\newtheorem{lemma}{Lemma}
\newtheorem{corollary}{Corollary}
\newtheorem{definition}{Definition}

\newtheorem{proposition}{Proposition}

\usepackage{verbatim}

\newcommand{\comm}[1]{}
\DeclareMathOperator\erfc{erfc}

\def\peq#1{\stackrel{\text{\scriptsize(#1)}}{=}}

\usepackage{cite}
\usepackage{colonequals}
\def\peq#1{\stackrel{\text{\scriptsize(#1)}}{=}}

\def\Pr{\mathbb{P}}
\def\Ex{\mathbb{E}}
\def\Var{\mathbb{V}}

%% file: intro.tex
\section{Introduction}

\subsection{Motivation}
For coverage, cellular operators deploy more base stations (BSs) in regions with severe signal decay, and vice versa, such that users at the cell boundaries achieve a sufficient and consistent signal strength. As a result, the spatial deployment of BSs and the large-scale propagation conditions are inherently correlated. In most works, this correlation is ignored, $i.e.,$ the BS deployment is assumed independent of the shadowing coefficients.

The first and only work that considers joint spatial and propagation modeling is \cite{guo2015jsp}, where the authors reverse engineer the path loss exponent (PLE) of the power-law path loss model from the BS locations. A fundamental assumption in \cite{guo2015jsp} is that the PLE inside each Voronoi cell is determined by the BS locations such that users at the cell edge receive an average power $P_0$ from their nearest BS. 
It is shown that under this assumption, the PPP yields almost the same success probability as the triangular lattice networks. However, there are a few drawbacks to that model. Firstly, the power-law path loss model is inherently an end-to-end model---the total path loss when a signal travels through multiple cells cannot be decomposed into per-cell path loss functions. {For instance,  for a signal that travels through two cells each with diameter $d$ and PLE $\alpha$, one can not decompose the total path loss, $(d+d)^{-\alpha}$ into the product of per-cell path losses $d^{-\alpha} d^{-\alpha}$.} Secondly, the assumption that the average power (over fading) $P_0$ is received by all users along the Voronoi cell edge is overly optimistic. In an actual deployment, this quantity is inevitably subject to variation.
And lastly, the coverage analysis in \cite{guo2015jsp} is limited to the spatial average, whereas the coverage used by operators is better captured by the meta distribution \cite{haenggi2016meta}.

{This work proposes a joint spatial and propagation model of cellular networks based on the directional radii of Poisson Voronoi cells. Specifically, our work assumes that the Poisson deployment of BSs results from the following BS placement method: BSs are deployed more densely in regions with severe signal attenuation and less densely in regions with more benign propagation conditions. In other words, the shape and size of the Voronoi cells reflect the underlying propagation conditions, which we reverse-engineer to devise a cell-dependent correlated shadowing model. To do so, it is necessary to study the cell shape and radii in the Poisson Voronoi tessellation (PVT). The contributions of the work are summarized as follows.}

\subsection{Contributions}
\begin{enumerate}
    \item  We characterize the shape and size of the Poisson Voronoi cells by introducing the notion of the directional radius in Voronoi tessellations.
    
    \item For the PVT, we derive the exact distributions of the directional radius in the zero-cell and the uniform-angled radius in the typical cell. The results reveal the asymmetry of Poisson Voronoi cells and also lead to a new approach of evaluating the mean cell areas. For cases without an explicit expression, simulation results and approximations are provided.
    
    \item We introduce and study a joint spatial-propagation (JSP) model for coverage-oriented cellular networks. We consider cell-dependent shadowing where the shadowing coefficients are conditionally log-normal random variables given the BS point process such that users at the cell edges receive an expected power $P_0$. Hence the JSP model ascribes the irregular deployment of base stations to an intelligent design by the operators, rather than to pure randomness, as is done in most of the literature.
    
    \item We show that the network performance depends critically on the variance of the received power along the cell boundary. While the PPP model (without shadowing or with independent shadowing) has been established as a pessimistic model for coverage-oriented deployments \cite{guo13spatial}, the SIR distribution of the JSP model for the PPP is close to that of the standard triangular lattice model (without shadowing) when the conditional variance (given the point process) is zero; as the variance increases, the performance of the JSP model for the PPP deteriorates to that of the standard PPP model.
\end{enumerate}
\subsection{Related Work}

{The shadowing coefficients introduced in this work are cell-dependent and  correlated. The correlation is due to the fact that in the PVT, nearby cells are correlated in shape and size and, in particular, in their directional radii.} In addition to \cite{guo2015jsp}, also relevant to this work are other models that consider correlated shadowing. 

One of the first correlated shadowing models is proposed in \cite{Gudmundson1991shadowing}, where for a fixed BS and a moving user with a constant velocity, the periodically sampled shadowing is a discrete process whose autocorrelation decays exponentially. Following \cite{Gudmundson1991shadowing}, the joint Gaussian distribution has been widely used to model correlated shadowing \cite{Szy2010shadowing}.
A correlated shadowing model with an intuitive physical interpretation is modeled and analyzed in \cite{Baccelli2015shadowing}, where the ``penetration loss'' depends on the number of obstacles (in this case, buildings) in the signal path.
The shadowing variance is another factor that significantly impacts the network performance for both independent and correlated shadowing models \cite{Blasz15Poisson},   \cite{Ross17stillPoisson}.
It is derived in \cite{Ross17stillPoisson} that for general BS processes satisfying a homogeneity constraint, if the shadowing correlation is ``moderate'' (decreasing fast enough in distance), the signal strengths converge to those in a PPP as the shadowing variance increases. {Based on \cite{Blasz15Poisson},   \cite{Ross17stillPoisson}, we obtain a Poisson convergence result for the JSP model.}

To facilitate the analysis of the JSP model,  we study two types of Poisson Voronoi cells and their radii: the zero-cell, which is the cell that contains the origin, and the typical cell. While it is known that the zero-cell has a larger mean volume than the typical cell \cite{gilbert1962random,hayen2002areas}, the directional radii characterize the shape of the two cells, which has not been studied before to the best of our knowledge. Related, the distance from the nucleus to a uniformly random location in the typical cell and the distance from the nucleus of the zero-cell to the typical location are studied in \cite{mankar2019distance}. User point processes are characterized based on the PVT in \cite{haenggi2017user}.
The distribution of the distance from the typical Voronoi edge/vertex location to its nearest Poisson point is given in \cite{muche_2005,baumstark_last_2007}, while \cite{calka_2002} derives the distribution of the radius of the largest disk included within the cells and the radius of the smallest disk containing the cells. Some gamma-type results are given in \cite{moller_zuyev_1996,zuyev_1999}.

\subsection{Layout}
The rest of the paper is organized as follows. Section II gives the definition of the directional radii of Voronoi cells and characterizes their distribution for the PVT. {Since the directional radii have applications beyond the JSP model, we are presenting a more comprehensive set of results than strictly necessary for the latter parts of the paper. In Section III, we introduce the JSP model and the performance metrics of interest. } Section IV provides the analysis of the JSP model and its comparison with other relevant models. Section V concludes the paper.

%% file: Distances.tex
\section{Directional Radii of Poisson Voronoi Cells}

\subsection{Definitions} 
 Let $\Phi\subset\R^2$ be a motion-invariant point process. To simplify the definitions of the cell radii, we first introduce the displaced typical cell and zero-cell such that the nucleus of the cells is at the origin $o$.

\medskip
{\em Typical cell.} Let
\[ \Phi^o\triangleq (\Phi \mid o\in \Phi) \]
and denote by $V(o)$ the Voronoi cell of $\Phi^o$ with nucleus $o$. {$V(o)$ is the typical cell in the Palm sense \cite{baumstark_last_2007}.} Let $z$ be a location chosen uniformly at random from $V(o)$
and let $(\|z\|,\zeta)$ be its polar coordinates. Next, define
\[ \tilde \Phi\triangleq\rot_{-\zeta} (\Phi^o) ,\]
where $\rot_u$ is a rotation around the origin by angle $u$, and denote the Voronoi cell of $\tilde\Phi$ with nucleus $o$ by $\tilde V(o)$. Let $D \triangleq \|z\|$ be the distance from the nucleus of the typical cell to the uniformly random location in the typical cell.

\medskip
{\em Zero-cell.}
Let $x_0\in\Phi$ be the closest point to the origin, $i.e.,$ $x_0=\argmin_{x\in\Phi} \{\|x\|\}$. Let $V_0$ be the Voronoi cell with nucleus $x_0$. By the definition of Voronoi tessellations, $V_0$ contains the origin. Letting $\varphi_0$ be the angle of $x_0$, define
\[ \tilde\Phi_0\triangleq \rot_{\pi-\varphi_0} (\Phi_{-x_0}) ,\]
where $\Phi_{y}$ is a translation of all points of $\Phi$ by $y$. This way, $o\in\tilde\Phi_0$. Let $\tilde V_0$ be the Voronoi cell of $\tilde\Phi_0$ with nucleus $o$.  Let $D_0 \triangleq \|x_0\|$.

\begin{figure}[t]
    \centering
    \begin{subfigure}{0.24\textwidth}
    \centering
    \includegraphics[width = \textwidth]{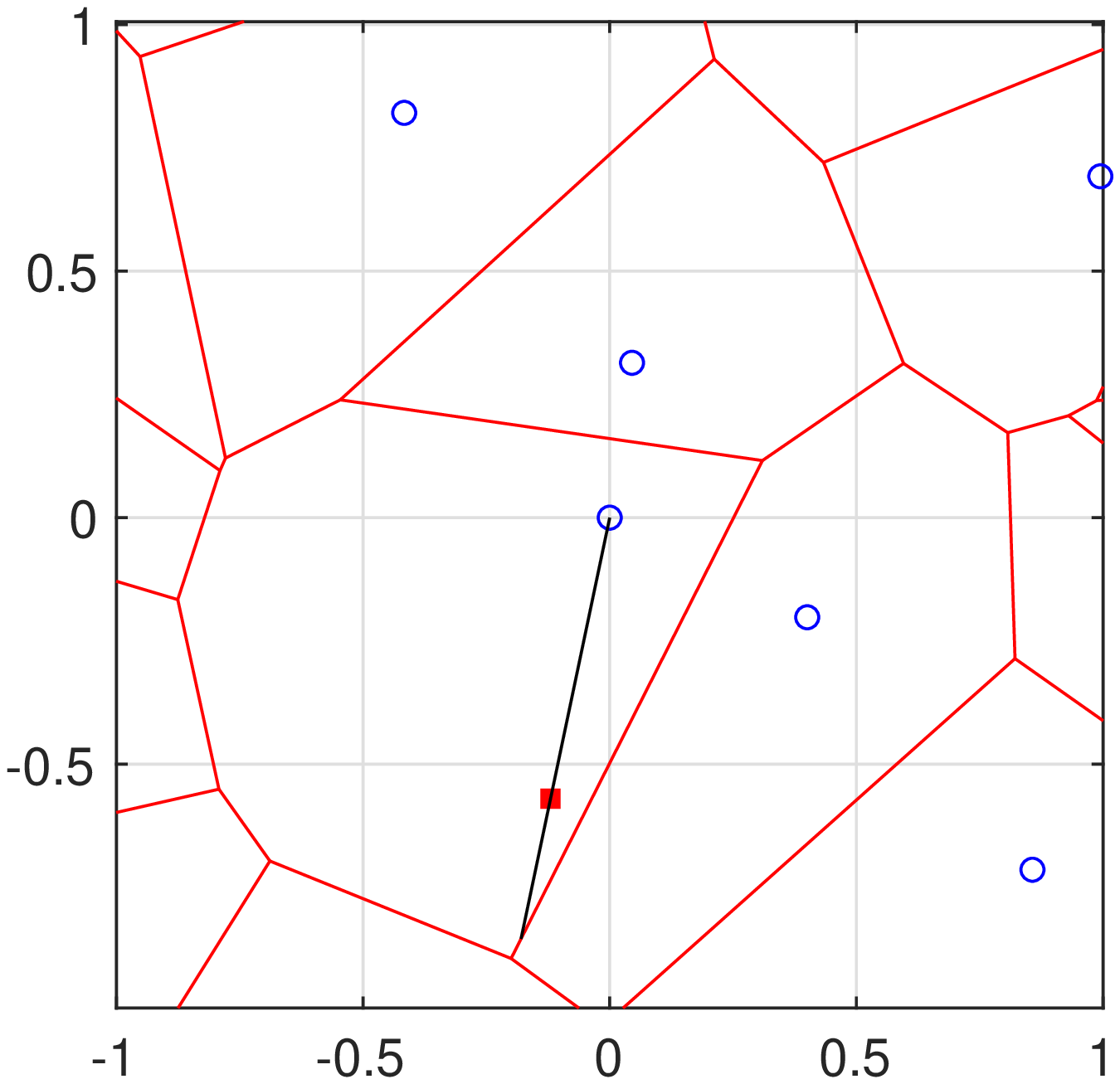}
 \caption{The typical cell. }\label{fig:hi-1}
    \end{subfigure}
     \begin{subfigure}{.24\textwidth}
     \centering
    \includegraphics[width = \textwidth]{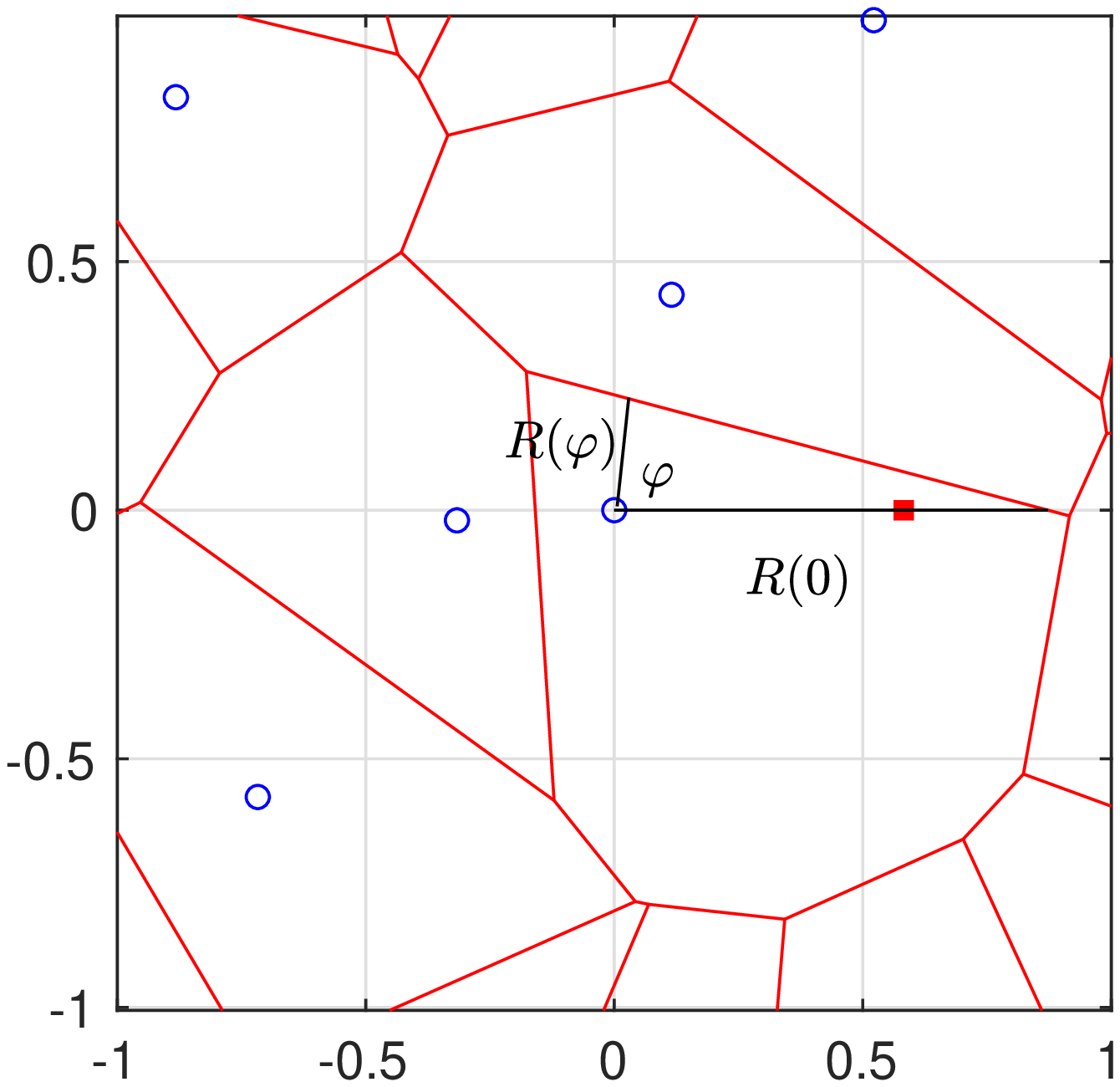}
     \caption{Rotated typical cell.}
    \end{subfigure}   
        \vspace{1em}

     \begin{subfigure}{.24\textwidth}
     \centering
    \includegraphics[width = \textwidth]{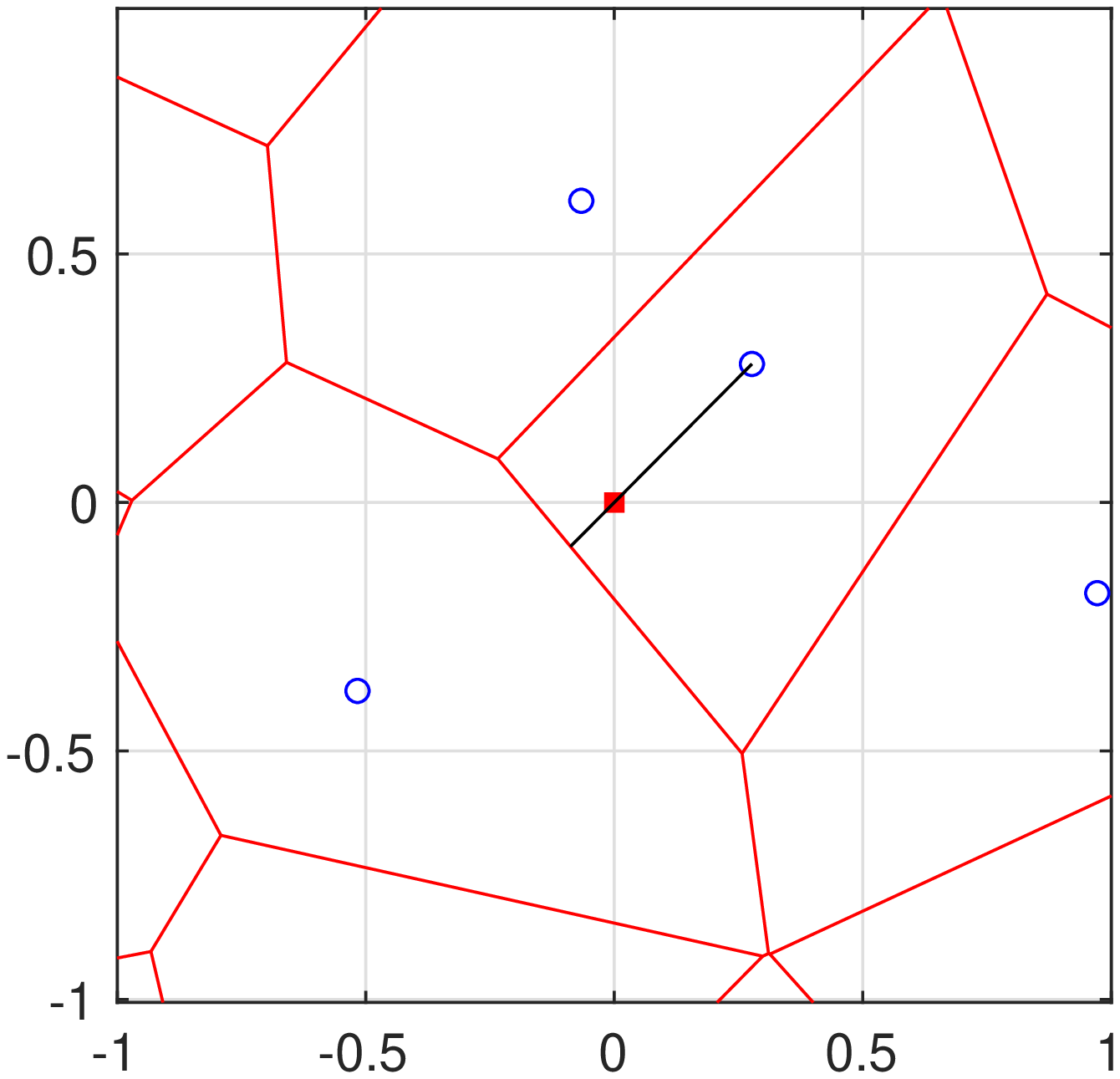}
     \caption{The 0-cell.}
    \end{subfigure}
     \begin{subfigure}{.24\textwidth}
     \centering
    \includegraphics[width = \textwidth]{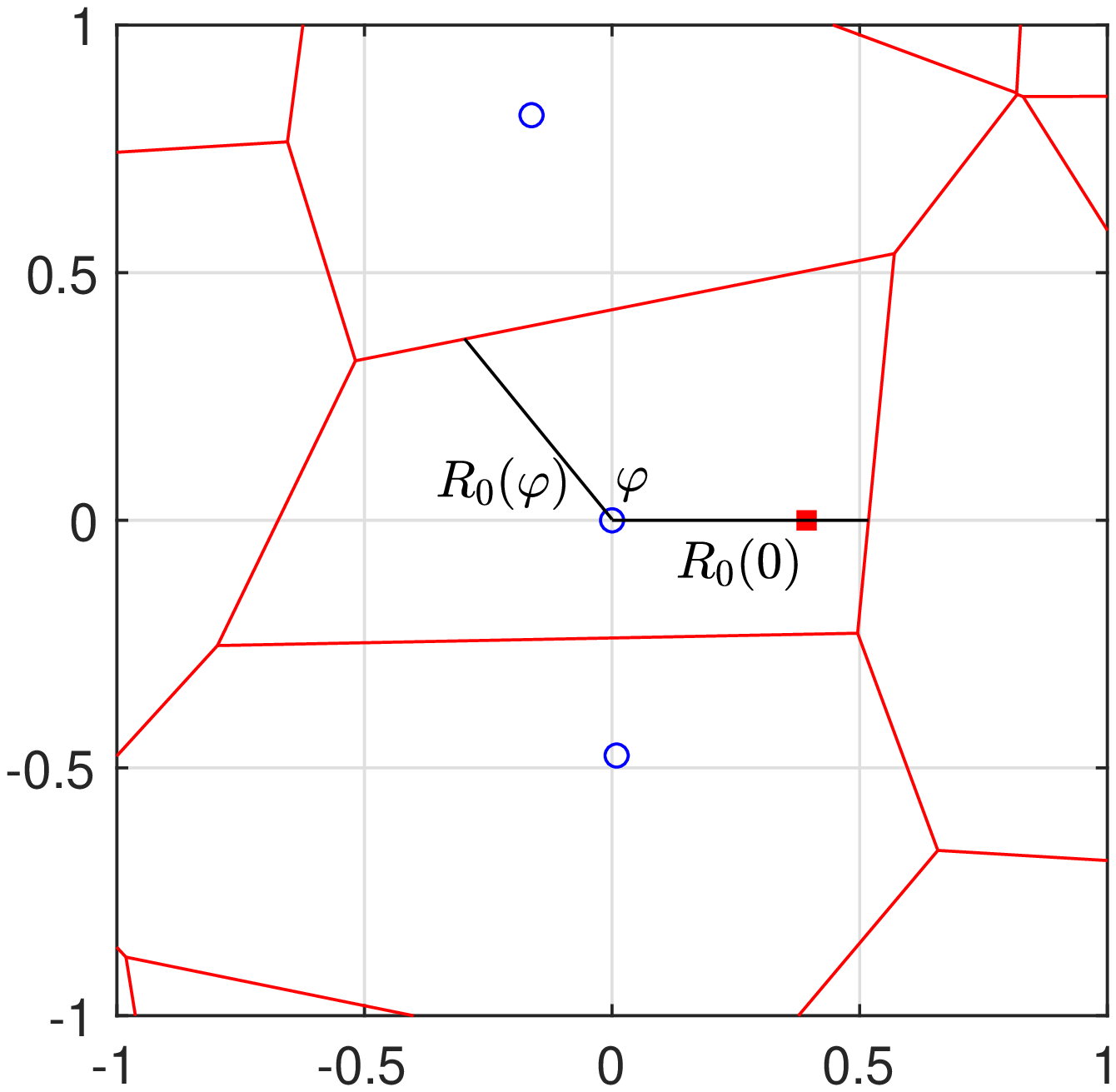}
     \caption{Rotated and displaced 0-cell.}
    \end{subfigure}  
    \caption{Illustration of the directional distances in the typical cell and the 0-cell of a PPP. The blue circles represent Poisson points and the red lines represent the Voronoi tessellations. In (a), the red square represents the uniform randomly distributed point $z$ in the typical cell. In (b), the cell is rotated such that $z$ resides on the positive x-axis. In (c) and (d), the red square represents the origin and the displaced origin.}
  \label{fig: illustration-rotation}
\end{figure}

\begin{definition}[Directional radius]
For $\varphi\in[0,2\pi)$, we define the directional radius $R(\varphi)$ to the boundary $\partial \tilde V(o)$ of the typical cell by
\[ (R(\varphi),\varphi) \in \partial \tilde V(o) \]
and the directional radius $R_0(\varphi)$ to the boundary $\partial \tilde V_0$ of the 0-cell by
\[ (R_0(\varphi),\varphi) \in \partial \tilde V_0. \]
\end{definition}
$(R(\varphi),\varphi)_{\varphi\in [0,2\pi)}$ parametrizes the boundary of the typical cell $\tilde V(o)$ in polar coordinates, and $R(0)$ is the distance
from the nucleus to the boundary in the direction of the randomly chosen point.
Similarly, $(R_0(\varphi),\varphi)_{\varphi\in [0,2\pi)}$ parametrizes the boundary of the 0-cell $\tilde V_0$ in polar coordinates, and $R_0(0)$ is the distance
from the nucleus to the boundary in the direction of the displaced origin, now at coordinates $(\|x_0\|,0)$. 
Fig. \ref{fig: illustration-rotation} shows realizations of the typical cell, the zero-cell and their displaced version when $\Phi$ is a Poisson point process.

The areas of the two cells are obtained as
\[ |\tilde V(o)|=\frac12 \int_0^{2\pi} R^2(\varphi)\dd\varphi \]
and
\[ |\tilde V_0|=\frac12 \int_0^{2\pi} R_0^2(\varphi)\dd\varphi, \]
respectively, and the mean areas follow as
\[ \E |\tilde V(o)|=\int_0^{\pi} \E(R^2(\varphi))\dd\varphi \]
and
\[ \E |\tilde V_0|=\int_0^{\pi} \E(R_0^2(\varphi))\dd\varphi, \]
where $|\cdot |$ is the Lebesgue measure in two dimensions. Integrating over $[0,\pi)$ is sufficient due to the symmetry of the distributions, $i.e.,$ $\E R(\varphi) \equiv \E R(-\varphi)$.

\begin{definition}[Uniform-angled radius]
We define the uniform-angled radius $\bar{R}$ to the boundary $\partial \tilde V(o)$ of the typical cell by
\[ \bar{R} \triangleq R(\Theta) \]
and the uniform-angled radius $\bar{R}_0$ to the boundary $\partial \tilde V_0$ of the 0-cell by
\[ \bar{R}_0 \triangleq R_0(\Theta)\]
where $\Theta$ is distributed as ${\rm{Uniform}}[0,2\pi].$
\end{definition}
Since $\Phi$ is motion-invariant, we may equivalently define $ \bar{R} \triangleq \|\partial V(o)\cap (\R^+,0) \|$ and  $ \bar{R}_0 \triangleq \|\partial V_0\cap (\R^+,0) \|$.

$R$ and $\bar{R}$ are related by
\begin{equation}
\E (\bar{R}^b)=\frac{1}{\pi} \int_0^{\pi} \E( R^b(\varphi))\dd\varphi, 
\end{equation}
and
\begin{equation}
\E (\bar{R}_0^b)=\frac{1}{\pi} \int_0^{\pi} \E( R_0^b(\varphi))\dd\varphi,
\end{equation}
for $b\in\mathbb{R}$. Again, integrating over $[0,\pi)$ is sufficient due to the symmetry.

\begin{lemma}
\label{lemma: |x_0|/r(x)}
For all point processes where $V(o)$ and $V_0$ exist and are finite almost surely, we have
\begin{equation}
\P(\|z\|/R(0) \leq t)=t^2, \quad t\in [0,1],
\end{equation}
and
\begin{equation}
\P(\|x_0\|/R_0(0) \leq t)=t^2, \quad t\in [0,1].
\end{equation}
\end{lemma}
\begin{proof}
For any point process, conditioned on $V(o)$, let $z$ be uniform randomly distributed in $V(o)$. The probability that $\|z\|/R(0)<t$ is the same as the probability that $z$ falls into the similar polygon of $V(o)$, with radius scaled by $t$ in all directions. This probability is equal to $t^2$ for any realization of $V(o)$. The same argument holds for the zero-cell.
\end{proof}

\begin{remark}
Lemma \ref{lemma: |x_0|/r(x)} holds for non-stationary point processes also,  where the typical cell is centered at the origin.
\end{remark}

\subsection{The Typical Cell of the PVT}
Let $\Phi\subset\mathbb{R}^2$ be a Poisson point process of intensity $\lambda$.
\begin{lemma}
\label{lemma: typical cell,unif angle}
The probability density function (pdf) of $\bar{R}$ is
\begin{equation}
f_{\bar{R}}(r)= 2\lambda\pi  r e^{-\lambda \pi r^2}.
\end{equation}
\end{lemma}
\begin{proof}
Due to the isotropy of the Poisson process, it is sufficient to consider $\bar{R} = \|\partial V(o)\cap (\R^+,0) \|$. The event that $\bar{R}$ is larger than $r$ happens if $b((\bar{R},0),r)$\footnote{The open  ball with center $(R,\phi)$ (in polar coordinates) and radius $r\geq0$ is denoted by $b((R,\phi), r)$.} contains no point. Thus, $\P(\bar{R}>r)= e^{-\lambda \pi r^2}.$
\end{proof}

\begin{remark}
The mean area of the typical cell follows as
\[ \E |V(o)|= \pi \E (\bar{R}^2)=\frac{1}{\lambda}.\]
\end{remark}
Recall that in \cite{hayen2002areas}, the mean area of the typical cell is obtained by using Robbin's formula \cite{robbins1944} and that for any fixed point $p=(r,\theta)$, $\P(p\in V(o)) = \exp(-\lambda\pi r^2)$,
$\E |V(o)|=\int_{\mathbb{R}^{2}} \P\left(p \in V(o)\right) \mathrm{d} p=\int_{0}^{2 \pi} \int_{0}^{\infty} \exp \left(-\lambda \pi r^{2}\right) r \mathrm{d} r \mathrm{d} \theta=\frac{1}{\lambda}.$
Our method and the method in \cite{hayen2002areas} for calculating the mean area are essentially the same, by observing that the event that $\bar{R}$ is larger than $r$ happens if and only if a fixed point $(r,0)\in V(o)$. Its probability does not depend on $\theta$.  The result for the mean area holds for arbitrary stationary point processes \cite{gilbert1962random}.
\begin{figure}[t]
    \centering
    \includegraphics[width = 0.8\columnwidth]{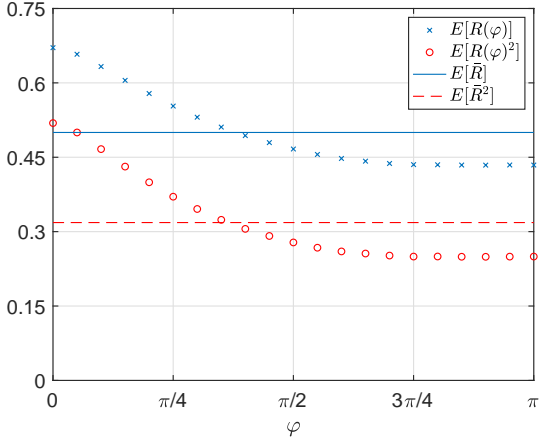}
    \caption{First two moments of the directional distances in the typical cell, $\lambda = 1$, via simulation. The mean and second moment of $\bar{R}$ (straight lines) are obtained via Lemma \ref{lemma: typical cell,unif angle}.}
    \label{fig:moment-typical}
\end{figure}

Fig. \ref{fig:moment-typical} shows the first two moments of the directional radius in the typical cell obtained via simulation. It is apparent that the cell is significantly larger in the direction of the randomly chosen point than in the opposite direction. $R(0)$ is on average about $55\%$ larger than $R(\pi)$.

\subsection{The 0-cell of the PVT}
Recall that $D_0=\|x_0\|$ is the distance from the nucleus of the 0-cell to the origin.

\begin{theorem}
\label{theorem: 0-cell,unif loc}
The joint pdf of $D_0, R_{0}(\varphi)$ for  $\varphi\in[0,\pi)$ is 
\begin{align}
\label{eq:joint-pdf-R_phi}
     f_{D_0,R_0(\varphi)}(x,y)
      & = 2\lambda \pi x \exp{\big(-\lambda\pi (x^2+y^2)+\lambda S(\varphi,x,y)\big)}\nonumber\\
      &\times\bigg(2\lambda\pi y -\lambda \frac{\partial S(\varphi,x,y)}{\partial y}\bigg),
\end{align}
for $x\geq 0,~y\geq 0$ when $\varphi\neq 0$, and for $y\geq x \geq 0$ when $\varphi= 0$, and 
\begin{align}
    S(\varphi,x,y) &= (\pi-\varphi)x^2-xy\sin{\varphi}\nonumber\nonumber\\
    &+(y^2-x^2)\arccos{\frac{y-x\cos{\varphi}}{\sqrt{x^2+y^2-2xy\cos{\varphi}}}}.
    \label{eq:intersection}
\end{align}
\end{theorem}
\begin{proof}
The event $R_0(\varphi)>y$ given $\|x_0\|=x$ is equivalent to there being no point in $b((y,\varphi),y)\setminus b((x,0),x) $. Hence
\begin{equation}
\label{eq:ccdf-D(phi)}
    \P(R_0(\varphi)>y\mid D_0 = x) = \exp{\big(-\lambda(\pi y^2 -  S(\varphi,x,y))\big)}.
\end{equation}
where $S(\varphi,x,y)$ in (\ref{eq:intersection}) is the area of the intersection of $b((x,0),x)$ and $b((y,\varphi),y)$, $i.e.,$ $S(\varphi,x,y) = |b((x,0),x) \cap b((y,\varphi),y)|$.

Hence the conditional pdf of $R_0(\varphi)$ given $D_0$ is
\begin{align}
    &f_{R_0(\varphi)\mid D_0}(y\mid x) \nonumber\\
   & = \exp{\big(-\lambda\pi y^2 +\lambda S(\varphi,x,y)\big)}\bigg(2\lambda\pi y -\lambda \frac{\partial S(\varphi,x,y)}{\partial y}\bigg).
\end{align}
From the void probability of the PPP we know that
\[f_{D_0}(x) = 2\lambda\pi x\exp(-\lambda \pi x^2).\]
Applying the Bayesian rule $ f_{D_0, R_0(\varphi) D_0}(x,y)= f_{R_0(\varphi)\mid D_0}(y\mid x) f_{D_0}(x)$ we obtain (\ref{eq:joint-pdf-R_phi}). \end{proof}
Fig. \ref{fig:intersection} illustrates the directional radius $R_0(\varphi)$ and the intersection region.

\begin{figure}[t]
    \centering
    \includegraphics[width = 0.8\columnwidth]{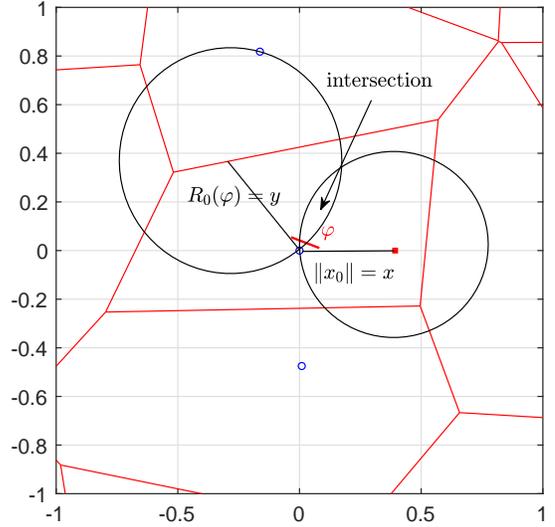}
    \caption{Illustration of the intersection between $b((x,0),x)$ and $b((y,\varphi),y)$ whose area is $S(\varphi,x,y)$.}
    \label{fig:intersection}
\end{figure}
\begin{remark}
Integrating (\ref{eq:joint-pdf-R_phi}) over $x$ we obtain the distribution for $R_0(\varphi),\varphi\in[0,\pi]$. A straightforward extension of Theorem \ref{theorem: 0-cell,unif loc} is the joint distribution of $R_0(\varphi_1),R_0(\varphi_2),D_0$ for $\varphi_1\in[0,\pi],\varphi_2\in[0,\pi]$, which involves the intersection of three open balls. Such an extension is useful when evaluating the second moment of $|\tilde{V}_0|$ but is omitted here.
\end{remark}

\begin{corollary}
\label{cor: pdf, R_0}
The pdf of $R_{0}(0)$ is
\begin{equation}
  f_{R_0(0)}(y)
   = 2(\lambda\pi)^2y^3\exp{(-\lambda \pi y^2)},
\end{equation}
and the pdf of $R_0(0) - D_0$ is 
\begin{equation}
f_{R_0(0) - D_0}(y) = \sqrt{\lambda}\pi \erfc{(y\sqrt{\lambda\pi})}.
\end{equation}
The pdf of $R_{0}(\pi)$ is
\begin{equation}
 f_{R_0(\pi)}(y) = 2\lambda\pi y\exp{(-\lambda\pi y^2)}.
\end{equation}
Further, $D_0$ and $ R_0(\pi)$ are independent and identically distributed (iid).
\end{corollary}

\begin{proof}
See Appendix A.
\end{proof}

From Corollary \ref{cor: pdf, R_0}, we obtain $\E(R_0(0))=3/(4\sqrt\lambda)$,
$\E(R_0(\pi))=1/(2\sqrt\lambda)$, and
$\E({R_0(0)-D_0})=1/(4\sqrt{\lambda}).$ Thus, $R_0(0)$ is on average exactly $50\%$ larger than $R_0(\pi)$.

The correlation coefficient of $R_0$ and $R_0(0)-D_0$ follows as 
\[ \rho_{R_0,R_0(0)-D_0} = \frac{8-3\pi}{\sqrt{12-3\pi}\sqrt{16-3\pi}} \approx -0.3462 .\]
Also, $\E((R_0(0)-D_0)/D_0)=1$, but $\E(R_0(0)-D_0)/\E(D_0)=1/2$.

\begin{corollary}
\label{theorem: 0-cell,unif angle}
The pdf of $\bar{R}_0$ is
\begin{equation}
f_{\bar{R}_0}(y) = \frac{1}{\pi}\int_{0}^{\pi} f_{R_0(\varphi)}(y) \dd    \varphi.
\end{equation}
\end{corollary}
\begin{proof}
Combine $\Theta\sim {\rm uniform}[0,2\pi]$ and Theorem \ref{theorem: 0-cell,unif loc}.
\end{proof}
Corollary \ref{theorem: 0-cell,unif angle} immediately leads to $\E\bar{R}_0 =  \E [\int_{0}^{\pi} R_0(\varphi) \dd    \varphi]/\pi = 0.5753/\sqrt{\lambda}$.

\begin{remark}
The mean area of the 0-cell is
\[ \E |\tilde V_0|=\int_0^{\pi} \E(R_0^2(\varphi))\dd\varphi = \frac{1.280176}{\lambda} . \]
Further,
\begin{equation}
\label{eq: c_phi}
    \E(R_0^2(\varphi)) \approx \frac{c(\varphi)}{\lambda\pi},
\end{equation}
where $c(\varphi)=1+\exp(-\varphi^{3/2})$, is a good approximation to the second moment of the directional radius. It gives a mean area of 
\[ 1+\frac{2\Gamma_{\rm inc}(2/3,\pi^{3/2})}{3\pi} \approx 1.2869 ,\]
where $\Gamma_{\rm inc}(a,z)=\int_0^z e^{-t} t^{a-1}\dd t$ is the lower incomplete gamma function\footnote{In Matlab, $\Gamma_{\rm inc}(2/3,\pi^{3/2})$ is
expressed as \texttt{gammainc(pi\^{ }1.5,2/3)*gamma(2/3)}.}. 
\end{remark}
\begin{figure}[t]
    \centering
    \includegraphics[width = 0.8\columnwidth]{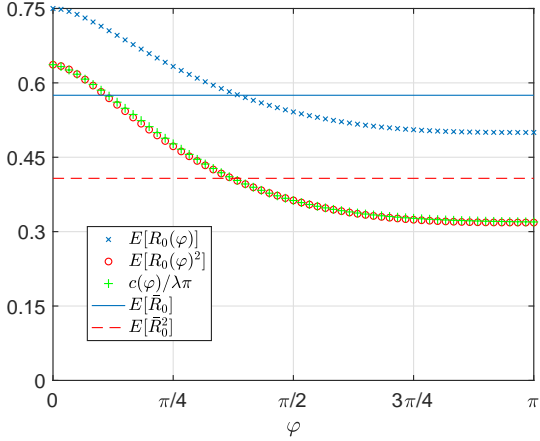}
        \caption{First two moments of the directional radius $R_0(\varphi)$ via Theorem \ref{theorem: 0-cell,unif loc} and the uniform-angled radius  $\bar{R}_0$ in the 0-cell, $\lambda=1$. The green curve, $c(\varphi)/\lambda\pi$, is given in (\ref{eq: c_phi}).}
 \label{fig:moment-0}
\end{figure}
Fig. \ref{fig:moment-0} shows the first two moments of $R_0(\varphi),~\varphi\in[0,\pi]$ and $\bar{R}_0$; it also shows the approximation $\E R^2_0(\phi)\sim c(\varphi)/(\lambda\pi)$ is quite good. This new approach for evaluating the mean area is easy to understand. By comparison, the existing approach is based on the first two moments of the area of the typical cell and the statistical relation between $V_0$ and $V(o)$  \cite{gilbert1962random,hayen2002areas}, which we discuss in the next subsection. 
 
\subsection{Relation of the Typical Cell and the 0-Cell}
Fundamentally, the typical cell and the zero-cell are related by \cite{MECKE19991645}
\begin{equation}
\E \big[f(V_0)\big] =\frac{ \E^o\big[ |V(o)| f(V(o))\big]}{\E^o\big[|V(o)|\big]},
\label{eq:V_0,V(o)}
\end{equation}
where $f$ is any translation-invariant non-negative function on compact sets, and $\E ^o$ denotes the expectation with respect to the Palm distribution \cite{haenggi2012stochastic}. In words, a translation-invariant statistic of the 0-cell is that of the typical cell weighted by volume (area in 2D). Letting $f(\cdot) = |\cdot|$, the mean area of the zero-cell is
\begin{equation}
\E \big[|V_0|\big] =\lambda \E^o\big[ |V(o)|^2\big].
\end{equation}
Using Robbin's formula, $
    \E^o\left(|V(o)|^2\right)=\int_{\mathbb{R}^{4}} \P\left(x_{0}, x_{1} \in V(o) \right) \mathrm{d} x_{0} \mathrm{d} x_{1}$ \cite{gilbert1962random} \cite{hayen2002areas}. It is apparent that the 0-cell is not just the typical cell enlarged by $28\%$. In fact, larger cells in the PVT are associated with being more circular and having more sides \cite{large-cell04}.
To compare the typical cell and the 0-cell, we consider the number of sides of the typical cell and the 0-cell, denoted by $N$ and $N_0$. We have $\E N_0 =\lambda \E^o[ |V(o)| N ]\geq \E N=6 $ due to the positive correlation between the area and number of sides of Poisson Voronoi cells \cite[Chap 9]{chiu2013stochastic}. 
Table \ref{tabel:mean} shows some mean values related to the typical cell and the 0-cell for $\lambda=1$. 

\begin{table*}
\centering
\begin{tabular}{ |c||c|c|c|p{3.6cm}|}
 \hline
 Cell Type & Number of Sides & Area & Directional Radius & Distance to $z$/the origin\\
 \hline
 Typical cell  & $\E N =6$  & $\E |V(o)|=1$ &    $\E R(0)= 0.67$ (*), $\E R(\pi)= 0.432$ (*)& $\E D = 0.447$ (*)\\
 Zero-cell & $\E N_0 = 6.41$ (*) & $\E |V_0| \approx 1.28$    & $\E R_0(0)= 0.75$, $\E R_0(\pi)= 0.5$ & $\E D_0=0.5$\\
 \hline
\end{tabular} 
\caption{Some mean values of the typical cell and the zero-cell in the PVT. Results obtained via simulations are marked by (*).}
\label{tabel:mean}
\end{table*}

\subsection{Gamma-Type Results}
We now compare our results with some known distributions. Corollary \ref{cor: pdf, R_0} shows that $\pi R_0^2(0)\sim \Gamma(2,\lambda)$; it is known that $\|x_1\|$, the distance between the origin and its second-nearest point, satisfies $\pi \|x_1\|^2 \sim \Gamma(2,\lambda)$ \cite{haenggi05distances}. Hence $R_0(0)$ and $\|x_1\|$ are identical in distribution. The explanation is as follows: for the PPP, a stopping set defined as the minimum disk containing $n$ Poisson points is $\Gamma(n,\lambda)$ distributed \cite{zuyev_1999}. Further, the probability that a point is covered by a stopping set does not depend on whether it is a point of the process or not. In our cases, both $\pi R_0^2(0)$ and $\pi \|x_1\|^2$ are defined by two Poisson points.

Denote the distance from the typical point on the edge to its closest Poisson point by $R_{\rm e}$ and the distance from the typical point on the Voronoi vertex to its closest Poisson point by $R_{\rm{v}}$. 
It is shown in \cite{muche_2005,baumstark_last_2007} that $\pi R_{\rm{e}}^2 \sim \Gamma({3}/{2},\lambda),$ and $\pi R_{\rm{v}}^2 \sim \Gamma({2},\lambda)$, which gives $f_{R_{\rm{e}}}(r)=4\lambda^{{3}/{2}}\pi r^2 e^{-\lambda\pi r^2},$
and $f_{R_{\rm{v}}}(r)=2(\lambda \pi)^{2} r^{3} e^{-\lambda \pi r^{2}}$. Hence $R_0(0)$ and $R_{\mathrm {v}}$ are identical in distribution.
\begin{figure}
    \centering
    \includegraphics[width = 0.8\columnwidth]{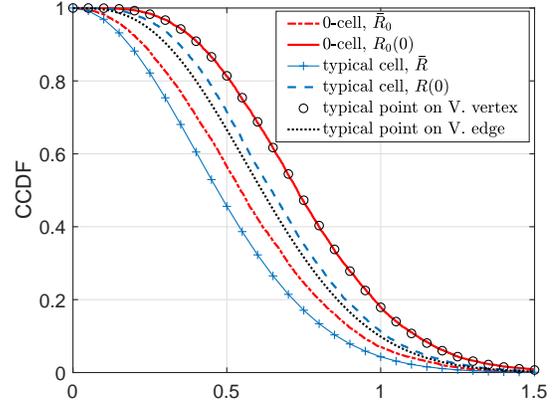}
    \caption{Distribution of distances in the 0-cell and the typical cell, $\lambda=1$.}
    \label{fig:distances dists}
\end{figure}
Fig. \ref{fig:distances dists} shows the  complementary cumulative distribution functions (ccdfs) of the distances given in Lemma \ref{lemma: typical cell,unif angle}, Theorems \ref{theorem: 0-cell,unif loc}, \ref{theorem: 0-cell,unif angle}, and the distributions of $R_{\rm e}$ and $R_{\rm v}$.

\subsection{Discussion and Impact of Cell Asymmetry}
From the results on the directional radii, it is apparent that the Poisson Voronoi cells are, quite surprisingly, rather asymmetric around their nucleus. We summarize them in the facts below.
\begin{figure}
    \centering
    \includegraphics[width = 0.8\columnwidth]{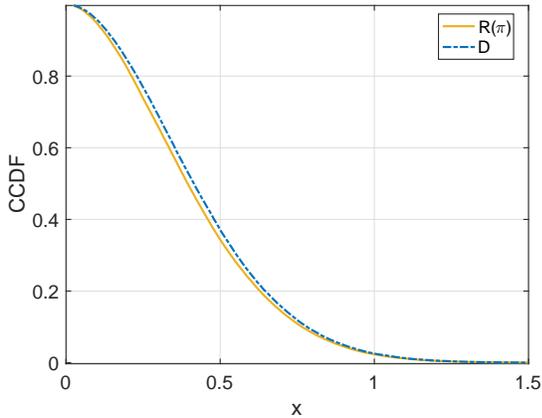}
    \caption{The distribution of $D$ and $R(\pi)$, $\lambda=1$, via simulation.}
    \label{fig:D, R(pi)}
\end{figure}
\begin{fact}
{For the zero-cell, the mean radius in the direction of the typical user is $50\%$ larger than the mean radius in the opposite direction, i.e., $\mathbb{E}\left(R_{0}(0)\right) / \mathbb{E}\left(R_{0}(\pi)\right)=3 / 2 .$ The typical user  is at the same distance as an edge user in
the opposite direction, since $R_{0}(\pi)$ and $D_{0}$ are iid. Further, we can infer from Fig. \ref{fig:moment-0} that about a quarter of edge users (those with $\varphi\geq3\pi/4$) are at essentially the same distance as the typical user. }
\end{fact}
\begin{fact}
{For the typical cell, numerical results from Table I suggest that $\E R(\pi)$ is 3\% smaller than $\E D$. The ccdf of $D$ and $R(\pi)$ are plotted in Fig. \ref{fig:D, R(pi)}, which shows that the two curves are almost identical. In the one-dimensional case where $\varphi\in\{0,\pi\}$, the distribution of $R(\pi)$, derived in Appendix B, is identical to the distribution of $D$, derived in \cite[Theorem 1]{mankar2019distance}. Further, we can infer from Fig. \ref{fig:moment-typical} that about a quarter of edge users (those with $\varphi\geq3\pi/4$) are at essentially the same distance as the uniformly random user.}
\end{fact}
In addition, the distance from the typical BS to the nearest edge location, $R_{\min}$,  is distributed as $ f_{R_{\min}}(r)=8\lambda\pi r e^{-4\lambda\pi r^2} $ as it is half the nearest-neighbor distance in the PPP. Since $\E (\pi R_{\min}^2)=1/4$, $3/4$ of the interior users are farther from the nucleus than the nearest edge user. And $\E (R_{\min})$ is only 37\% of the mean distance in the direction of the uniformly random user.

These facts may prompt us to rethink some assumptions that are generally made, such as the claim that edge users necessarily suffer from low signal strength. Also, care is needed when evaluating the performance of non-orthogonal multiple access (NOMA) schemes, especially if “cell-center” refers to a user located uniformly at random in the cell and “cell-edge” refers to a user located uniformly at random on the edge of the cell. In this case, simply pairing a cell-center user as the strong user and an edge user as the weak one may be quite inefficient, since the edge user may be closer to the BS than the ``cell-center” user. Conversely, if “cell-center” and “cell-edge” are defined based on relative distances between serving and interfering base stations \cite{mankar20NOMA,Feng19SIRgain}, then a “cell-edge” user may actually be quite far from the edge of the cell. A potential model to pair users for Poisson Voronoi cells is to select a “cell-center” user uniformly at random inside the cell, and select an edge user whose angle differs only slightly from that of the “cell-center” user. This increases the likelihood of significant channel gain difference between users and thus increases the NOMA gain. An alternative model that guarantees the intended ordering of strong and weak user is to place the two randomly in the in-disk of the cell and then order them \cite{Ali19tcom}.

%% file: sys-model.tex
\section{A Joint Spatial-Propagation Model for Cellular Networks}

{In coverage-oriented cellular networks, it is natural to assume that the operator uses a deployment method where BSs are spaced more densely in regions with severe signal attenuation and less densely in regions with more benign propagation conditions. In this section, we assume that the BS locations result from such a deployment procedure. Consequently, we introduce the JSP model which reverse-engineers the underlying cell-dependent shadowing characteristics from the shape and size of the Voronoi cells. For the Poisson deployment, the Voronoi cell radii distributions are provided in the last section. We refer to the JSP model for the PPP as the JSP-PPP model. }
\subsection{System Model}
Let $\Phi\subset \mathbb{R}^2$ be a stationary point process with intensity $\lambda$ modeling BS locations.  The typical user is located at the origin $o$ without loss of generality. We assume all BSs are active and transmit with unit power. For $x\in\Phi$, denote by $h_x$ and $K_x$ the power of the small-scale iid Rayleigh fading with unit mean and the large-scale shadowing between $x$ and the origin, respectively. The power-law path loss model is considered, $i.e.,$ $\ell(x)=\|x\|^{-\alpha}$, where $\alpha>2$ is a constant. {Note that this propagation model applies to a low-density high-power BS deployment, which is usually well-planned. The framework can be generalized to a dense small cell networks setting by accounting for the LoS/NLoS effect with a multi-slope LoS/NLoS path loss model, in which case the BS density plays a more critical role. For instance, see \cite{LoS}. }

 Let $\{x_i\}_{i\in\mathbb{N}_0}$ be the point process ordered by the distance to the origin: \(x_0\triangleq\argmin_{x\in\Phi}\{\|x\|\}\) and so on. Let $r(x)$ be the distance from $x\in\Phi$ to its Voronoi cell edge oriented towards the typical user. Note that $r(x_0)\equiv R_0(0)$, which is the zero-cell radius in the direction of the typical user in Section II. Fig. \ref{fig:JSP-illu} shows a realization of $\Phi$ and the corresponding $r(x_0), r(x_1), r(x_2)$. By the construction of the Voronoi cells, $r(x_0)\geq \|x_0\|$ and $r(x_i)\leq 
\|x_i\|,~i\geq1$.  
\begin{definition}[Cell-dependent shadowing]
In cell-dependent shadowing, for given $\Phi$, $\{K_x\}_{x\in\Phi}$ are conditionally independent log-normal random variables such that the expected large-scale path loss from $x$ to its Voronoi cell boundaries is $P_0$, $i.e.,$
\begin{equation}
    \Ex[K_{x} r(x)^{-\alpha} \mid \Phi] = P_0.
    \label{eq:JSP_def}
\end{equation}
We denote by $\mu_x$ and $\sigma_x$ the mean and standard deviation of $\log(K_x)$ conditioned on $\Phi$, and we fix $\sigma_x \equiv \sigma\geq0$ for $\forall x \in\Phi$.  \end{definition}
\begin{figure}[t]
    \centering
    \includegraphics[width = 0.5\textwidth]{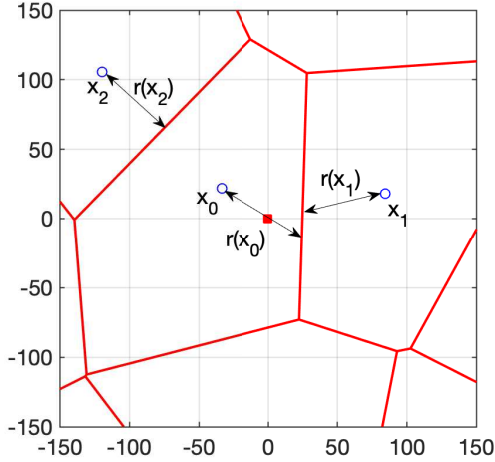}
    \caption{An illustration of the JSP model with cell-dependent shadowing. Blue circles are BS locations generated from a PPP with $\lambda = 3.5\times10^{-5}$. Red lines are the Voronoi tessellation. The typical user denoted by the red square is located at the origin. $r(x_i)$ is the length of the black line segment, which is the cell radius of $x_i$ oriented towards the typical user.}
    \label{fig:JSP-illu}
\end{figure}
\begin{remark}
Cell-dependent shadowing introduces dependence between shadowing and cell radius (determined by BS geometry). For the Poisson deployment, the shadowing coefficients are correlated because nearby cells in the PVT are correlated in shape and size and, in particular, in their directional radii. Intuitively, cells in proximity are shaped by some common points. In cell-dependent shadowing for a point pattern $\phi$ ($e.g.$, a realization of a PPP), the shadowing coefficients are independent (usually not identical) log-normal random variables.
\end{remark}

\begin{remark}
The model in \cite{guo2015jsp} assumes $r(x)^{-\alpha(x)} = P_0$, which captures the cell-dependent signal propagation through the PLE. It is sensitive to $P_0$ and $\lambda$ due to the singularity of the path loss model. Our model avoids its deficiency while generalizing several models in the literature: if in  (\ref{eq:JSP_def}), $P_0 = r(x)^{-\alpha}$ (instead of a constant), we retrieve the iid shadowing model in \cite{Blasz15Poisson}; if further $\sigma =0$, we retrieve the traditional model without shadowing (or constant shadowing) in \cite{andrews2011tractable}.  The log-normal model is commonly used for shadowing and allows us to compare this work with previous models. 
\end{remark}

\begin{remark}
{Note that the shadowing from an interfering BS $x$ to the typical user is assumed to only be related to the cell radius $r(x)$, and $\sigma_x\equiv \sigma$ is fixed for all BSs. This is a simplification\footnote{Such a simplification is common in the literature, $e.g.,$ often the shadowing coefficients from all BSs are modeled as identically distributed.} as the shadowing may occur along the signal path outside the cell, and more remote BSs may have a larger shadowing variation. Nevertheless, the assumption enables an average minimum received power at all cell edges, which is the primary concern for coverage. Further, it is expected that the power-law path loss is the dominating large-scale effect for remote BSs.}
\end{remark}

For the cell-dependent shadowing, $\sigma$ captures the variation of $K_x\|r(x)\|^{-\alpha}$ around $P_0$.  
\subsubsection{$\sigma = 0$}
For $\sigma = 0$, $\{{K_x}\}_{x\in\Phi}$ is a deterministic function of  $\Phi$. In this case, we have
\begin{equation}
    K_x = P_0r(x)^{\alpha}.
\end{equation}
Users located at the Voronoi cell edge of $x$ receive a constant signal power $P_0$ (averaged over small-scale fading) from $x$. {This corresponds to a scenario where operators have access to precise terrain and propagation data and the BS layout is optimized for coverage.}

\subsubsection{$\sigma > 0$}
For $\sigma > 0$, the shadowing in the JSP model is doubly random such that the power averaged over small-scale fading at the cell edge  fluctuates around $P_0$. In this case, we have
\begin{equation}
   \Ex[ K_x\mid \Phi] = P_0r(x)^{\alpha}.
\end{equation}
{This corresponds to a scenario where operators have imprecise terrain and propagation data or where the BS deployment is suboptimal for coverage.} Given $\Phi$, we have $ \exp(\mu_x+\sigma^2/2) = P_0r(x)^{\alpha},$
which yields $\mu_x = \log(P_0r(x)^{\alpha})-\sigma^2/2.$ Depending on whether $\sigma=0$ or $\sigma>0$, $\{{K_x}\}_{x\in\Phi}$ is either a deterministic function of $\Phi$ or is a set of random variables correlated with $\Phi$. From the expression of $\mu_x$, the correlation diminishes as $\sigma$ increases.

We consider the strongest-BS association throughout this paper, $i.e.,$ the typical user is served by the BS with the strongest signal averaged over small-scale fading. Denote the serving BS by $x=\argmax_{y\in\Phi}\{ {K_y\|y\|^{-\alpha}}\}$. The signal-to-interference ratio (SIR) is
\begin{equation}
{\rm{SIR}} \triangleq \frac{S}{I}  = \frac{h_xK_x\|x\|^{-\alpha}}{\sum_{y\in{\Phi\setminus\{x\}}}h_y K_y\|y\|^{-\alpha}}.
\end{equation}
\subsection{Performance Metrics}
We focus on the following three performance metrics.

\subsubsection{Asymptotic Gain}
The success probability is defined as $p_{\rm{s}}(\theta) \triangleq \Pr(\rm{SIR} > \theta),~\theta>0$. For models with iid Rayleigh fading, it is shown in \cite{haenggi2014mean} that
\begin{equation}
    \label{eq:asymp-equi-MISR}
   1-p_{\rm{s}}(\theta) \sim {\rm{MISR}} \ \theta,\quad \theta\to0,
\end{equation}
where $A(t)\sim B(t)$ means the limit of their ratio goes to 1, and the $\rm{MISR}$ (mean interference-to-signal ratio) is defined as\footnote{Shadowing is not considered in the model and definition of the MISR in \cite{haenggi2014mean}. But it is straightforward to extend the definition of the MISR to include shadowing.}
\begin{align}
      \nonumber{ \rm{MISR}}& \triangleq \Ex[I/\Ex_{h}[S]]\\ \nonumber&=\Ex\bigg[\sum_{y\in\Phi\setminus\{x\}}\frac{K_y\|y\|^{-\alpha}}{K_x\|x\|^{-\alpha}}\bigg].
\end{align}

Thus, we can compare the asymptotics of the success probabilities for different models by simply calculating the ratio of their MISRs. Throughout this paper, we use the standard PPP model as the baseline for comparison, where $\rm{MISR}_{PPP} = 2/(\alpha-2)$ \cite{haenggi2014mean}. Let $G$ denote the asymptotic gain. We have
\begin{equation}
    G = \frac{\rm{MISR}_{PPP}}{\rm{MISR}}.
\end{equation}
\subsubsection{SIR Meta Distribution}
For ergodic point processes, the SIR meta distribution \cite{haenggi2016meta} gives the fraction of users that achieve an SIR $\theta$ with a reliability higher than $x$, which is a more fine-grained performance metric than $p_{\rm{s}}(\theta)$. It is defined as $
    \bar{F}_{P_{\rm{s}}}(\theta,x)\triangleq \Pr(P_{\rm{s}}(\theta)>x),~ x\in[0,1]$, where $
    P_{\rm{s}}(\theta) \triangleq \Pr({\rm{SIR}}>\theta\mid \Phi,\{K_{y}\}_{y\in\Phi})$ is the conditional success probability. In words, $P_{\rm{s}}(\theta)$ is the reliability of the typical link under small-scale fading while the large-scale propagation (shadowing and path loss) is given. 
    For Rayleigh fading, the conditional success probability is
\begin{equation}
\begin{split}
    \nonumber P_{\rm{s}}(\theta) & \triangleq \Pr({\rm{SIR}}>\theta\mid \Phi, \{K_{y}\}_{y\in\Phi})\\
    \nonumber & = \Ex\bigg[\exp{\bigg(-\theta\sum_{y\in\Phi\setminus\{x\}}h_y\frac{K_{y}\|y\|^{-\alpha}}{K_{x}\|x\|^{-\alpha}}\bigg)}~\Bigr|~ \Phi,\{K_{y}\}_{y\in\Phi}\bigg]\\
     & \peq{a}  \prod_{y\in\Phi\setminus\{x\}}\frac{1}{1+\theta (\|x\|/\|y\|)^\alpha K_{x}/K_{y}}.\\
\end{split}    
\end{equation}
Step (a) follows from the iid exponential distribution of $h_x, x\in\Phi$. 
The $b$-th moment of the conditional success probability is
\begin{equation}
  M_b(\theta) = \Ex[P_{\rm{s}}(\theta)^b],~b\in\mathbb{C}.
\end{equation}
Note that $p_{\rm{s}}(\theta)\equiv M_1(\theta)$. 

\subsubsection{Path Loss Point Process}
We define the path loss point process\footnote{It is also referred to as the ``propagation process'' in \cite{Blasz15Poisson} or the ``signal spectrum'' in \cite{Ross17stillPoisson}.} for a general BS point process $\Phi$ to be \(\Pi\triangleq\) \(\left\{ \|x\|^\alpha/K_x \right\}_{x \in \Phi}\). The path loss point process, introduced in \cite{haenggi2008plpp}, characterizes the received signal strengths (averaged over small-scale fading) from all transmitters in the network from the viewpoint of the typical user. This notion helps establish equivalence between the performance of networks when their path loss point processes have the same distribution. To avoid a colocated BS and user, we assume no BS is located at the origin. 

\subsection{Relevant Results}
In the standard models, the shadowing is a constant, $i.e.,$ $K_x \equiv 1,~x\in\Phi$. The large-scale path loss depends only on the BS locations. The nearest BS  provides the strongest signal. It is known that for the standard PPP, $
     M_b(\theta) = {1}/{_2F_1(b,-\delta;1-\delta;-\theta)}$ \cite{haenggi2016meta}, $b\in\mathbb{C}$, where $_2F_1$ is the Gauss hypergeometric function, and $\delta\triangleq2/\alpha$.  The asymptotic gain $G$ captures the SIR gap due to BS deployment. For instance, the standard triangular lattice has an approximately 3.4 dB asymptotic SIR gain over the standard PPP for $\alpha=4$ \cite{haenggi2014mean}.

     In the iid log-normal shadowing model \cite{Blasz15Poisson}, $\{K_x\}_{x\in\Phi}$ are iid, and $\log K_x \sim \mathcal{N}(\mu,\sigma^2)$, $\mu=-\sigma^2/2$, so that $\E K_x\equiv1$. It is shown in \cite{Blaszczyszyn2010} that the path loss point process $\Pi$ for a PPP with iid shadowing is an inhomogeneous PPP. Thus, under the strongest-BS association, the iid log-normal shadowing model for the PPP performs exactly the same as the (baseline) PPP.
     Further, \cite{Blasz15Poisson} shows that when $\sigma\to\infty$ in the iid log-normal shadowing model, the path loss point process of any deterministic/stochastic BS point processes converges to that of a PPP, given the point process satisfies a mild homogeneity constraint. Remarkably, \cite{Ross17stillPoisson} proves that this conclusion also holds for moderately correlated shadowing.

%% file: analysis_JSP.tex
\section{Performance Analysis of the Joint Spatial-Propagation Model}
In this section, we analyze the performance of the JSP-PPP model. We focus on the distribution of the serving signal, shadowing distribution/correlation, the asymptotic SIR gain, the SIR meta distribution, and finally the path loss point process. We first introduce the lemma below.
\begin{lemma}
\label{lemma: r(x_i)/|x_i|}
For a Poisson point process with intensity $\lambda$, the ccdf of $r(x_i)/\|x_i\|$, $i\geq1$ is
\begin{equation}
    \P(r(x_i)/\|x_i\| > t) = (1-t^2)^i,\quad t\in[0,1],
    \label{eq: r(x_i)/|x_i|}
    \end{equation}
    and the ccdf of $r(x_i)$ is
    \begin{equation}
\Pr(r(x_i)>t) = \exp(-\lambda\pi t^2).
\label{eq: r(x_i), i>=1}
    \end{equation}
\end{lemma}
\begin{proof}
Recall that $x_i$ is the $i+1$-th closest point to the origin. Let $\Phi (b(o,r))$ denote the number of points in $\Phi$ falling in the disk of radius $r$ centered at the origin.
For $t\in[0,1]$,
\begin{align}
  \nonumber  &\P(r(x_i)/\|x_i\| > t) \\
  \nonumber & = \E \P(r(x_i) > \|x_i\|t \mid \|x_i\|)\\
   \nonumber & \peq{a} \E \P( \Phi (b(o,\|x_i\|t)) = 0 \mid \Phi (b(o,\|x_i\|)) = i)\\
  \nonumber & \peq{b} \E \left(\frac{\|x_i\|^2-{\|x_i\|}^2{t}^2}{\|x_i\|^2}\right)^i\\
  \nonumber & = (1-t^2)^i.
\end{align}
Step (a) holds since the probability of having no point inside a disk only depends on the radius of the disk, not on the disk center. Step (b) follows from the property of the PPP, where conditioned on $\|x_i\|$, the $i$ points are distributed uniformly at random in $b(o,\|x_i\|)$. Combining (\ref{eq: r(x_i)/|x_i|}) with the distribution of $\|x_i\|$ \cite{haenggi05distances} we obtain the ccdf for $r(x_i)$, $i\geq1$, in (\ref{eq: r(x_i), i>=1}).
\end{proof}

\subsection{The Serving Signal}
For $\sigma = 0$, the nearest BS $x_0$ provides the strongest signal. Hence $ \Ex_h[S] = K_{x_0}\|x_0\|^{-\alpha}$.
We have
\begin{align}
\nonumber\Pr(K_{x_0} \|x_0\|^{-\alpha} > t)
& = \Pr(P_0 r(x_0)^\alpha\|x_0\|^{-\alpha} > t)\\\nonumber
& = \Pr(\|x_0\|/r(x_0)<({P_0}/{t})^{1/\alpha})\\
& \peq{a} P_0^\delta{t}^{-\delta},\quad t\geq P_0,
\label{eq:p_T}
\end{align}
where $t\geq P_0$ due to the minimum received power constraint. Step (a) follows from Lemma \ref{lemma: |x_0|/r(x)}. The distribution of $K_{x_0} \|x_0\|^{-\alpha}$ does not depend on the intensity or distribution of $\Phi$, and it is equal to the distribution of the signal power in a disk where the received power at the cell edge is $P_0$. In other words, for the serving signal, the JSP model turns any irregular cell shape into a disk. For the standard model, (\ref{eq:p_T}) can be shown to hold asymptotically \cite[Lemma 7]{ganti2016asymptotics}. 

Fig. \ref{fig:strongest-signal dist.} shows the distribution of $\Ex_h[S]$ for the JSP-PPP model and (\ref{eq:p_T}) with $\sigma = 0,~ \lambda = 10^{-2},~P_0 = 0.5$, and $\alpha = 4$. Fig.  \ref{fig:strongest-signal dist in Tri.} shows that we can use $(P_0/\theta)^{-\delta}$ to approximate the distribution of the signal from the nearest BS in the standard triangular lattice network, which is not surprising considering that hexagonal cells and circular cells are similar in shape.  The intensity of the triangular lattice in Fig.  \ref{fig:strongest-signal dist in Tri.} is scaled such that $P_0 = (\lambda\pi)^{1/\delta}$ for a fair comparison. Note that unlike in Poisson networks, there is a minimum average received power in lattices determined by the intensity of the point process.

We further obtain the tail of the ccdf of $S$ as follows. 
\begin{figure}[t]
    \centering
     \begin{subfigure}[t]{0.75\columnwidth}
    \includegraphics[width = \linewidth]{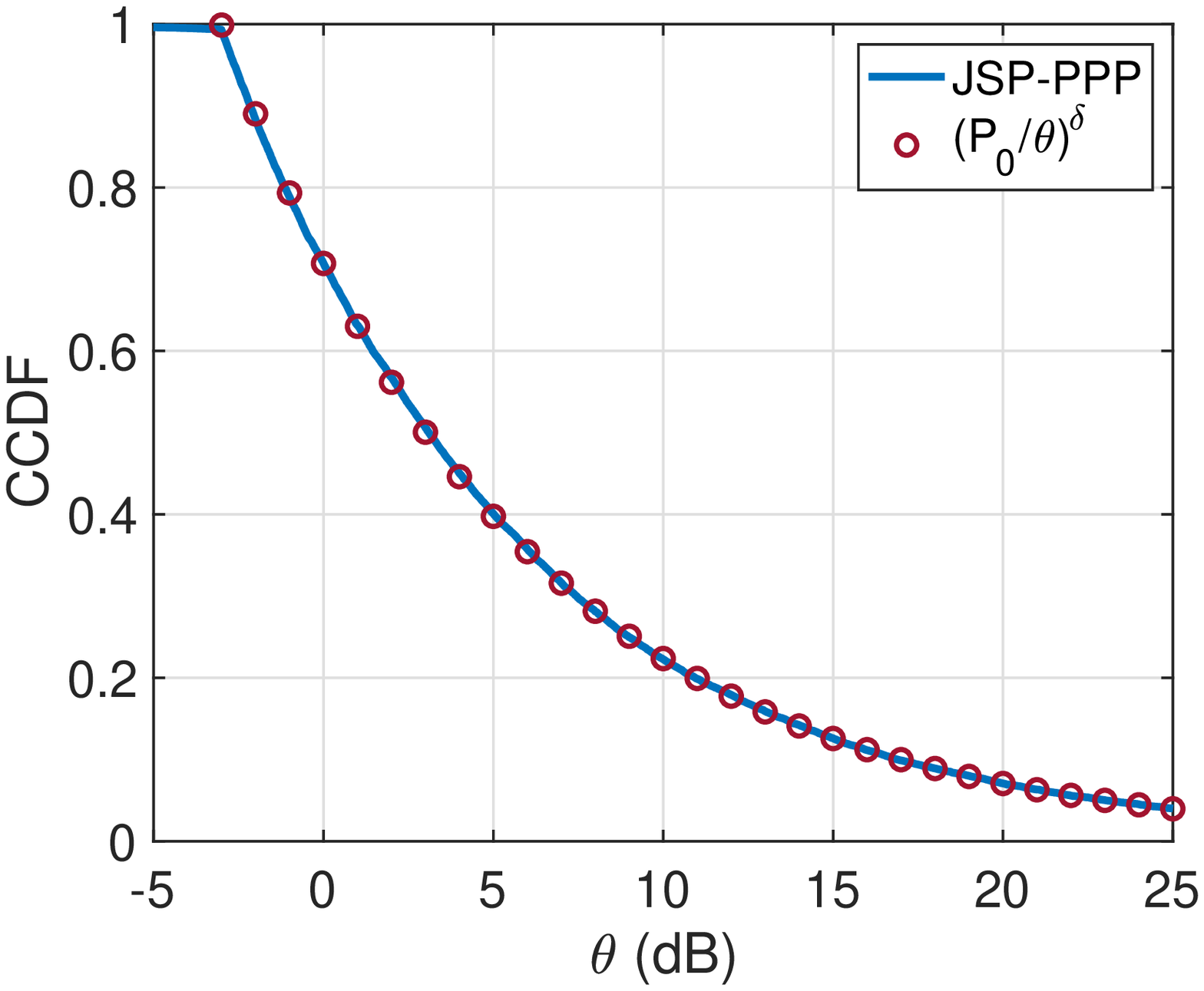}
     \caption{JSP-PPP, $\sigma = 0$.}
     \label{fig:strongest-signal dist.}
    \end{subfigure}
    \vspace{1em}

     \begin{subfigure}[t]{0.75\columnwidth}
    \includegraphics[width =  \linewidth]{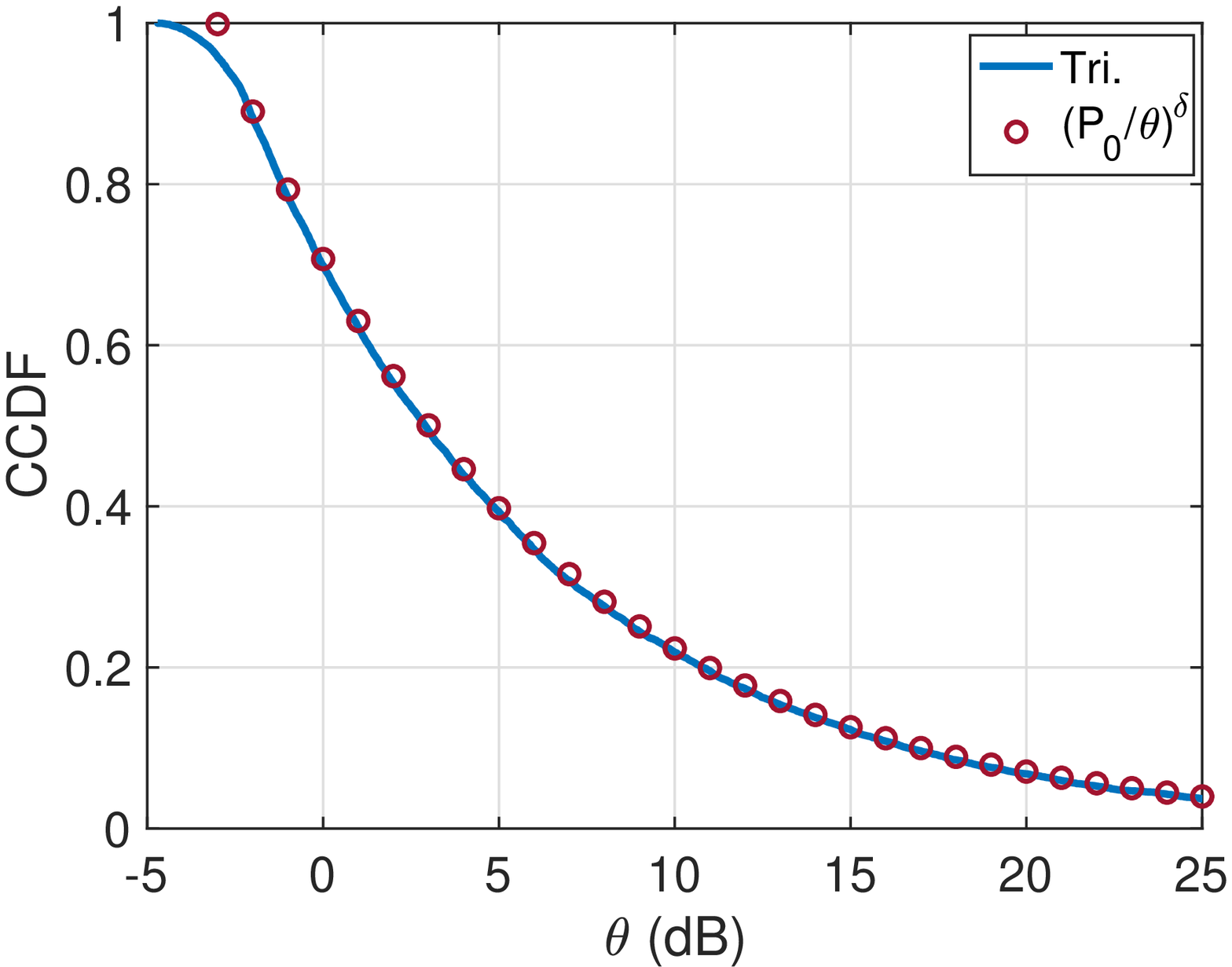}
     \caption{Triangular lattice, $\lambda = P_0^\delta/\pi=0.225$.}
     \label{fig:strongest-signal dist in Tri.}
    \end{subfigure}
    \caption{The distribution of $\Ex_{h}[S]$. $P_0 = 0.5$, $\alpha=4.$}
\end{figure}

\begin{lemma}
For the JSP model with any BS process and $\sigma=0$,
\begin{equation}
    \Pr(S > t)\sim P_0^\delta\Ex(h^\delta){t}^{-\delta}, \quad t\to\infty.
    \label{eq:p_S}
\end{equation}
\end{lemma}
\begin{proof}
\begin{align}
\Pr(S > t)\nonumber
& = \Pr(P_0h_{x_0} r(x_0)^\alpha\|x_0\|^{-\alpha} > t)\\\nonumber
& = \Pr(\|x_0\|/r(x_0)<({P_0h_{x_0}}/{t})^{1/\alpha})\\\nonumber
&\sim P_0^\delta\Ex(h^\delta){t}^{-\delta},\quad t\to\infty.
\end{align}
\end{proof}
In \cite[Lemma 7]{ganti2016asymptotics}, it is shown that for the standard model, the tail of the ccdf of the desired signal strength for all stationary point processes is $\mathbb{P}(S>t) \sim \lambda \pi \mathbb{E}(h^{\delta}) t^{-\delta},~t\to\infty.$ If we let 
\begin{equation}
\nonumber P_0 = (\lambda\pi)^{1/\delta},
\end{equation}
we obtain the same tails. Intuitively, if we could ``pack'' the space with congruent disks, we would have $r^{-\alpha} = (1/\lambda\pi)^{-\alpha/2} = (\lambda\pi)^{1/\delta} = P_0$.

For $\sigma>0$, the serving BS $x=\argmax_{y\in\Phi}\{ {K_y\|y\|^{-\alpha}}\}$.
\begin{align}
   \nonumber &\Pr(K_{x} \|x\|^{-\alpha} <t )\\
   & = \Pr(K_{y} \|y\|^{-\alpha}<t, y\in\Phi )\\
   \nonumber & =  \Ex \prod_{y\in\Phi} \Pr(K_y<\|y\|^{\alpha}t\mid \Phi)\\
   \nonumber & = \Ex \prod_{y\in\Phi} \bigg(\frac{1}{2}+\frac{1}{2}\operatorname{erf}\left[\frac{\log t\|y\|^\alpha-\log P_0r(y)^\alpha +\sigma^2/2}{\sqrt{2} \sigma}\right]\bigg).
\end{align}

\subsection{Shadowing Coefficients}

    \subsubsection{Distribution}
For $\sigma =0$, the shadowing coefficient from any BS is a deterministic function of the cell radius of that BS oriented towards the origin. For the serving cell,
\begin{align}
\Pr(K_{x_0}\geq t) 
    & = \exp{(-\lambda\pi t^\delta P_0^{-\delta})}(1+\lambda \pi t^\delta P_0^{-\delta}),
    \label{eq:K_x_0}
\end{align}
and
\begin{align}
\Pr(K_{x_i}\geq t)
    &= \exp{(-\lambda\pi t^\delta P_0^{-\delta})},
    \label{eq:K_x_1}
\end{align}
which follow from the distribution of $r(x_0)$ and $r(x_i),i\geq1$, in Theorem \ref{theorem: 0-cell,unif loc} and Lemma \ref{lemma: r(x_i)/|x_i|}, respectively.

Based on (\ref{eq:K_x_0}) and (\ref{eq:K_x_1}), $\Ex [K_{x_0}]  = P_0(\lambda\pi)^{-\alpha/2}\Gamma(\alpha/2+2).$ $\Ex [K_{x_i}] = P_0(\lambda\pi)^{-\alpha/2}\Gamma(\alpha/2+1),~i\geq1$. Denoting by $\Var K_{x_i}$ the variance of $K_{x_i}$, we have $\Var K_{x_0} = P_0^2(\lambda\pi)^{-\alpha}(\Gamma(\alpha+2) - \Gamma(\alpha/2+2)^2)$
 and $\Var K_{x_i} = P_0^2(\lambda\pi)^{-\alpha}(\Gamma(\alpha+1) - \Gamma(\alpha/2+1)^2),~i\geq 1$.
Fig. \ref{fig:shadowing-dist} shows the ccdfs for $K_{x_0}$ and $K_{x_1}$. Fig. \ref{fig:shadowing-moment} shows the mean and standard deviation of $K_{x_0}$ and $K_{x_1}$ versus $\alpha$ based on (\ref{eq:K_x_0}) and (\ref{eq:K_x_1}). $K_{x_0}$ statistically dominates $K_{x_i},i\geq1$, since $r(x_0)$ statistically dominates $r_{x_i},i\geq1$.

\begin{figure}[t]
    \centering
     \begin{subfigure}{0.8\columnwidth}
    \includegraphics[width = \linewidth]{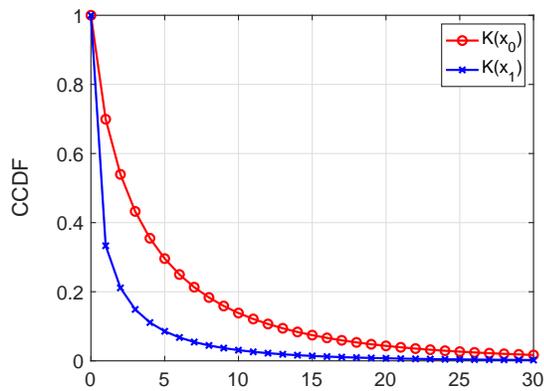}
    \caption{Ccdfs of the shadowing per (\ref{eq:K_x_0})  and (\ref{eq:K_x_1}). }
    \label{fig:shadowing-dist}
    \end{subfigure}
\vspace{1em}

\begin{subfigure}{0.8\columnwidth}
    \includegraphics[width = \linewidth]{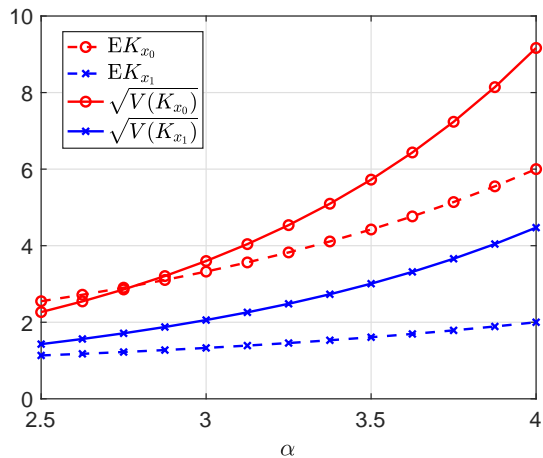}
    \caption{Mean and standard deviation. $P_0 = (\lambda \pi)^{\alpha/2}.$}
       \label{fig:shadowing-moment}
       \end{subfigure}
 \caption{Ccdfs, means, and standard deviations of $K_{x_0}$ and $K_{x_1}$, $\sigma=0$, $\lambda=1$, $\alpha=4$.}
\end{figure}

For $\sigma>0$, the ccdf of $K_{x_i}$ is
\begin{equation}
\begin{split}
    \Pr(K_{x_i}\geq t )  &= \Ex \Pr(K_{x_i}\geq t \mid r(x_i))\\\nonumber
    & =   \frac{1}{2}\Ex \erfc{\left(\frac{\log{t}-\log(P_0r(x_i)^\alpha)+\sigma^2/2}{\sqrt{2}\sigma}\right)},
\end{split}
\end{equation}
where the distribution of $r(x_0)$ is given in Theorem \ref{theorem: 0-cell,unif loc} and the distribution of $r(x_i),i\geq1$, is given in Lemma \ref{lemma: r(x_i)/|x_i|}. $\sigma$ appears in both the denominator and numerator inside of the erfc function. When $\sigma\to\infty$, the impact of $r(x_i)$ diminishes.

\subsubsection{Correlation}
We consider two types of shadowing correlation. The first type is the correlation between the shadowing coefficients from two BSs to the typical user. The second type is the correlation between the shadowing coefficient and the directional radius of a cell. In the proposed JSP model, these two types of correlation are inherently related, $i.e.,$ the correlation between shadowing is induced by the correlation between cell radius. If the BS deployment is modeled by a point pattern ($i.e.,$ deterministic point process), only the second type of correlation exists.

Let $P_0=1$ for simplicity. The correlation coefficient between the shadowing coefficients $K_x,K_y$ (from BS $x,y\in\Phi$)
 is
\begin{align}
   \nonumber \rho_{K_x,K_y}& = \frac{\Ex[K_xK_y]-\Ex K_x\Ex K_y}{\sqrt{\Var{K_x}}\sqrt{\Var{K_y}}},
\end{align}
where $\Ex[K_xK_y] = \Ex[r(x)^\alpha r(y)^\alpha],~\Ex K_x=\Ex[r(x)^\alpha]$, and $\Var{K_x}=\exp(\sigma^2)\Ex r(x)^{2\alpha}-(\Ex r(x)^\alpha)^2$. {As the distance between two BSs $x, y$ increases, the correlation between $r(x)$ and $r(y)$ vanishes. Hence the locality of the shadowing correlation is preserved. }
Obviously, $\rho_{K_x,K_y} \leq {\rho_{r(x)^\alpha,r(y)^\alpha}}$, and the equality holds when $\sigma=0$. Further, $\rho_{K_x,K_y}$ decreases with $\sigma$.  
For $\sigma\to\infty$, $\rho_{K_x,K_y} \to 0$ for any $x\neq y$. 

The correlation between $K_x$ and $r(x)^\alpha$ 
 is
\begin{align}
  \nonumber  \rho_{K_x,r(x)^\alpha}& = \sqrt{\frac{\Var (r(x)^\alpha)}{\Var{K_x}}},
\end{align}
where again, $\Var{K_x}=\exp(\sigma^2)\Ex r(x)^{2\alpha}-(\Ex r(x)^\alpha)^2$. $\rho_{K_x,r(x)^\alpha}=1$ for $\sigma=0$. For $\sigma\to\infty$,  $\rho_{K_x,r(x)^\alpha}\to 0$.

\subsection{Asymptotic Gain}
\begin{figure}[t]
    \centering
     \begin{subfigure}{0.8\columnwidth}
    \includegraphics[width = \linewidth]{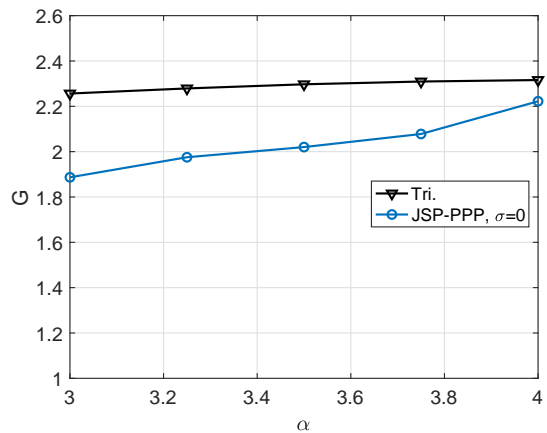}
    \caption{The asymptotic gain for triangular lattices and the JSP-PPP with $\sigma =0$.}
    \label{fig:deployment-gain}
    \end{subfigure}
    \vspace{1em}
    
     \begin{subfigure}{0.8\columnwidth}
    \includegraphics[width = \linewidth]{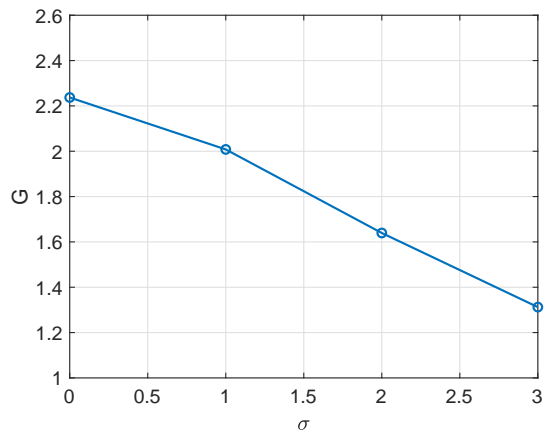}
    \caption{The asymptotic gain for the JSP-PPP with $\sigma = 0,1,2,3$, for $\alpha=4$. }
    \label{fig: JSP_G_sigma}
    \end{subfigure}
    \caption{The asymptotic gain of the JSP-PPP model relative to the standard PPP model. {Note that by definition, the standard PPP model yields $G\equiv 1$.}  }
\end{figure}

The MISR of the JSP-PPP model is \(
   { \rm{MISR}} = \Ex[{\sum_{y\in\Phi\setminus\{x\}}K_y\|y\|^{-\alpha}}/{K_{x}\|x\|^{-\alpha}}]\nonumber
\), which is independent of $P_0$ and $\lambda$. For $\sigma = 0,$ we have $\mathrm{MISR} = \Ex[\sum_{y\in\Phi\setminus\{x\}}{r(y)^\alpha}{\|y\|^{-\alpha}}/{r(x)^\alpha}{\|x\|^{-\alpha}}]$. The correlation between $r(x),r(y), \|x\|,\|y\|$ makes the calculation of the MISR involved. Hence we use simulations to study the impact of $\alpha$ and $\sigma$. 
Fig. \ref{fig:deployment-gain} shows the asymptotic gain (relative to the standard PPP model) for the standard triangular lattice model and the JSP-PPP with $\sigma=0$, which increases with $\alpha$. 
Fig. \ref{fig: JSP_G_sigma} shows the asymptotic gain $G$ for the JSP-PPP decreases with $\sigma$. 
As discussed in the last subsection, increasing $\sigma$ decreases the correlation between shadowing and cell radius as well as the correlation between the shadowing coefficients.  Eventually, as $\sigma\to\infty$ the JSP-PPP model reverts to the PPP with iid log-normal shadowing. 

\subsection{SIR Meta Distribution}
\begin{figure}
    \centering
    \includegraphics[width = 0.95\columnwidth]{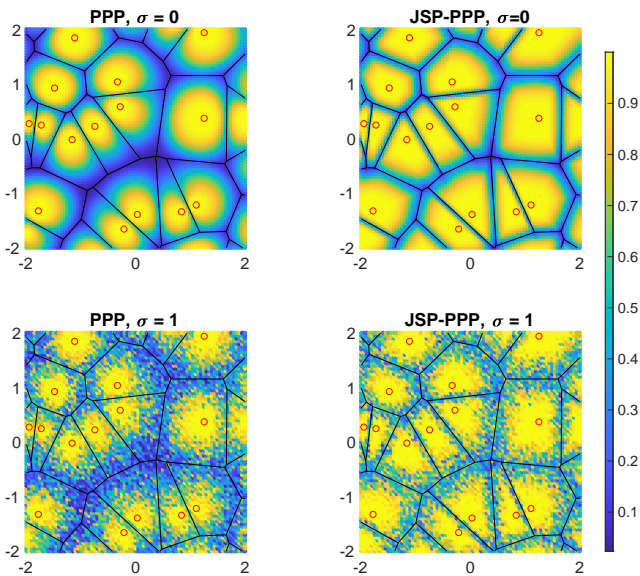}
     \caption{Illustration of the conditional success probabilities for $\theta=1$ for the PPP with iid log-normal shadowing and the JSP-PPP with $\sigma = 0,1$ under the strongest-BS association, $\lambda=1$, $\alpha=4$. }
    \label{fig:shadowing-cell-illu}
\end{figure}

Fig. \ref{fig:shadowing-cell-illu} shows how the conditional success probabilities with a fixed $\theta=1$ are distributed for the PPP with iid log-normal shadowing and the JSP-PPP model with the strongest-BS association. For $\sigma=0$, the region where $P_{\rm{s}}(\theta)>0.8$ appears elliptical around the nucleus for the PPP; in contrast, for the JSP-PPP, the region where $P_{\rm{s}}(\theta)>0.8$ is enlarged and adapts to the cell shape almost perfectly. For $\sigma = 1$, both regions are blurred due to the shadowing variance.

Fig. \ref{fig:md-sigmaneq0} shows the simulation results for the SIR meta distribution of the JSP-PPP model with fixed reliabilities. The meta distribution for the (standard) triangular lattice and the (standard) PPP model are plotted for comparison. Under the strongest-signal association, the meta distribution decreases with $\sigma$, shifting the curve towards that of the PPP.

\begin{figure}[t]
    \centering
    \begin{subfigure}{0.8\columnwidth}
    \includegraphics[width = \linewidth]{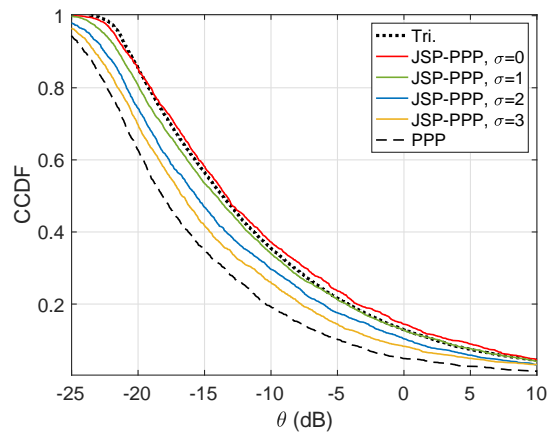}
     \caption{$\bar{F}_{P_{\rm{s}}}(\theta,x),~x=0.99$.}
    \end{subfigure}
        \vspace{1em}

     \begin{subfigure}{0.8\columnwidth}
    \includegraphics[width =  \linewidth]{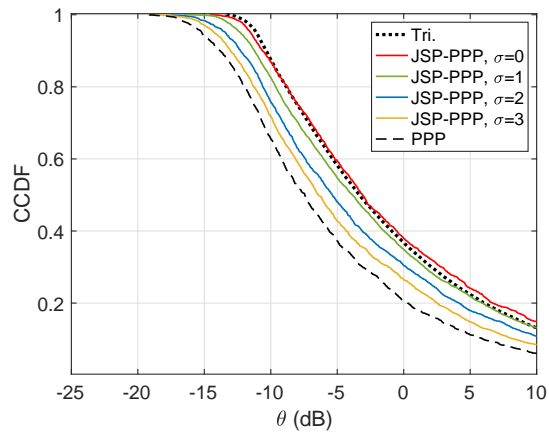}
    \caption{$\bar{F}_{P_{\rm{s}}}(\theta,x),~x=0.9$.}
    \end{subfigure}
    \caption{The SIR meta distribution for the JSP-PPP model with $\sigma = 0,1,2,3$ and $x=0.9,0.99$, $\alpha=4$. The black dashed and dotted curves denote the meta distribution for the PPP and the triangular lattice (without shadowing), respectively. }
    \label{fig:md-sigmaneq0}
\end{figure}

\begin{figure}[t]
    \centering
     \begin{subfigure}{0.8\columnwidth}
    \includegraphics[width = \linewidth]{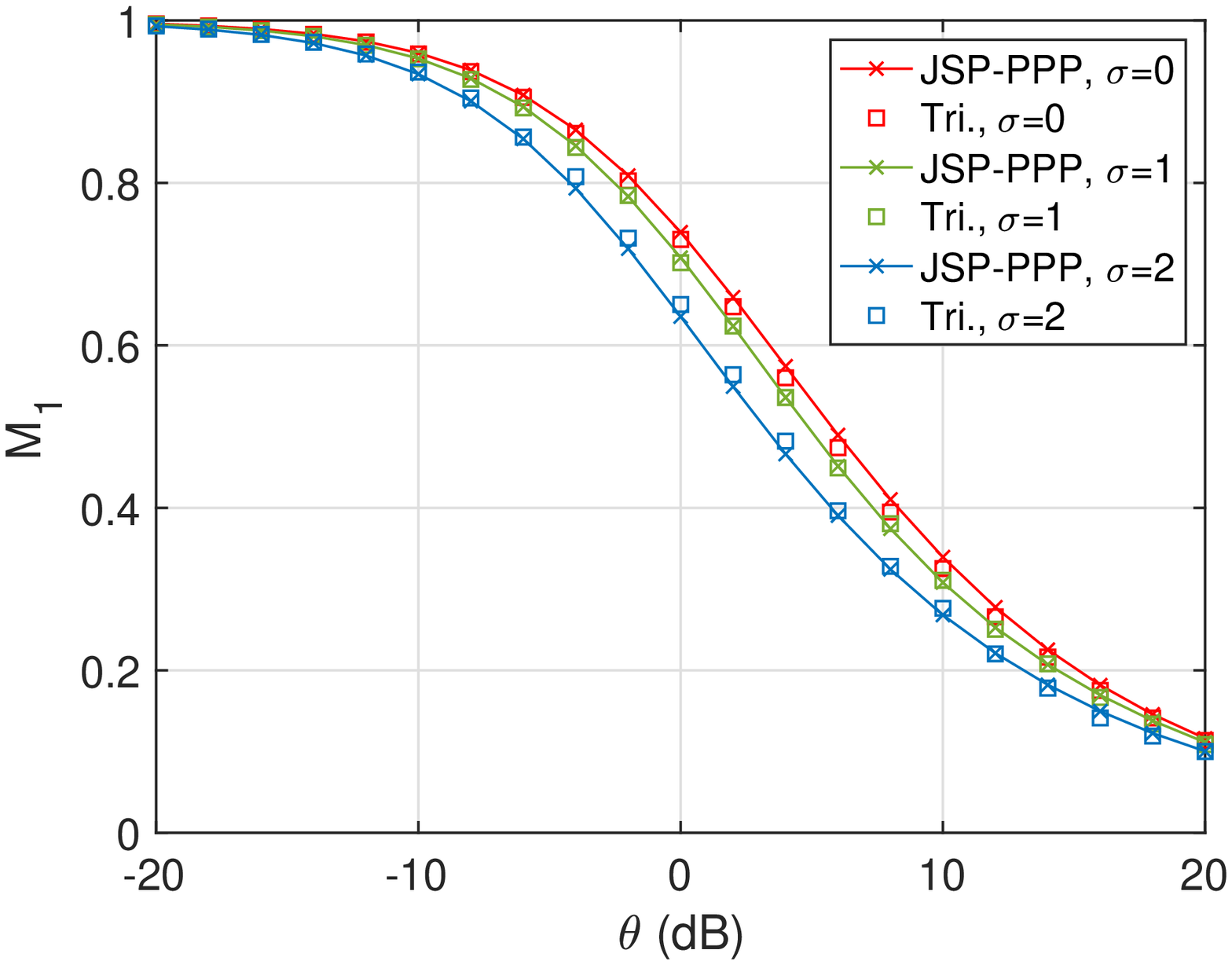}
     \caption{$M_1(\theta)$.}
     \label{fig:m1}
    \end{subfigure}
\vspace{1em}

     \begin{subfigure}{0.8\columnwidth}
    \includegraphics[width =  \linewidth]{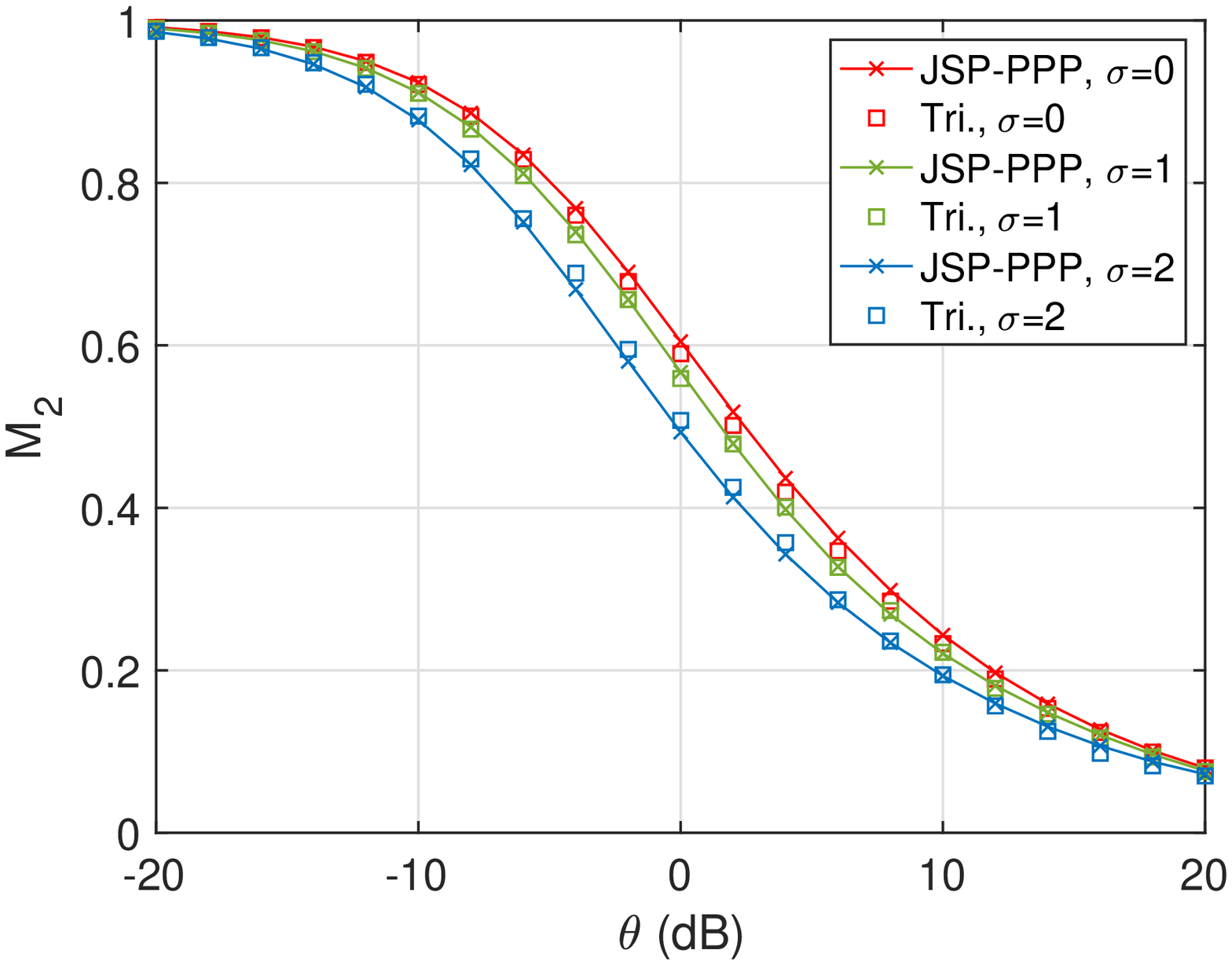}
     \caption{$M_2(\theta)$.}
     \label{fig:m2}
    \end{subfigure}
    \caption{First two moments of the conditional success probability for JSP-PPP and the triangular lattice with iid log-normal shadowing. $\sigma=0,1,2$, $\alpha=4$.}
    \label{fig:m1,m2}
\end{figure}

Fig. \ref{fig:m1,m2} plots the first two moments of the conditional success probability for the JSP-PPP model and the triangular lattice with iid log-normal shadowing. Both moments are approximately the same for a set of different values of $\sigma$. Hence the meta distribution for the JSP-PPP model is close to that of the triangular lattice, since the first two moments generally lead to a good approximation of the meta distribution \cite{haenggi2016meta}.

\subsection{Convergence of the Path Loss Point Process}
The path loss point process of the JSP model for a point pattern $\phi$ is \(\Pi=\) \(\left\{\|x\|^\alpha/ K_{x}\right\}_{x \in\phi}\). In this subsection,  we show that the path loss point process of the JSP model for any realization of the PPP converges to that of a PPP as $\sigma\to\infty$. First we recall a result from \cite{Blasz15Poisson}.

\begin{proposition}{\cite{Blasz15Poisson}}
\label{Proposition: convergence}
For any deterministic and locally finite collection of points $\phi\subset \mathbb{R}^2$ without a point at the origin, let the shadowing coefficients, $\{K_x\}_{{x\in\phi}}$, be iid log-normal random variables with $\E K_x = 1$
and $\Var(\log(K_x)) =\sigma^2$. If there is a constant \(0<\lambda<\infty\) such that as \(t \rightarrow \infty\)
\begin{equation}
\frac{\phi\left(b(o,t)\right)}{\pi t^{2}} \rightarrow \lambda,
\label{eq:homo}
\end{equation}
then the path loss point process $\Pi$ after rescaling by $(\E K_x^\delta)^{1/\delta} = \exp{(-\sigma^2(1-\delta)/2)}$ converges weakly as \(\sigma \rightarrow \infty\) to that of the PPP on \(\mathbb{R}^{+}\) with intensity measure $
\Lambda([0,t))=\lambda\pi t^2.$
\end{proposition}
The rescaling of $\Pi$ by $(\E K_x^\delta)^{1/\delta}$ is necessary to obtain a non-zero intensity measure as $\sigma\to\infty$. Now, when $\phi$ be a realization of the PPP, we have the convergence of the path loss point process for the JSP model as follows.

\begin{lemma}
    The path loss point process of the JSP model for any realization of the PPP after rescaling by $P_0\exp{(-\sigma^2(1-\delta)/2)}$ converges weakly as \(\sigma \rightarrow \infty\) to that of the PPP on \(\mathbb{R}^{+}\) with intensity measure $
\Lambda([0,t))= t^2. $  
\end{lemma}
\begin{proof}
We first show that the JSP model for a point pattern $\phi$ can be viewed as the iid log-normal shadowing model in Proposition \ref{Proposition: convergence} with a modified point pattern $\hat{\phi}$. Then we show that when $\phi$ is a realization of the PPP, its modified BS point pattern $\hat{\phi}$ satisfies the convergence criterion.

For the JSP model,
$\{K_x\}_{x\in\phi}$ are independent but not necessarily identically distributed log-normal random variables such that $\E K_x = P_0 r(x)^\alpha$
and $\Var(\log(K_x)) =\sigma^2.$ We have \[\Pi =
 {\left\{ x\in\phi\colon \frac{\|x\|^\alpha}{r(x)^\alpha}\frac{1}{K_x/r(x)^\alpha}\right\}} = \left\{x\in\hat{\phi}\colon\frac{\|x\|^\alpha}{\hat{K}_x} \right\},\]
where $\hat{\phi} \triangleq \{ x\in\phi \colon{x}/{r(x)}\}$ and $\hat{K}_x \triangleq K_x/r(x)^\alpha$. Now $\{\hat{K}_x\}_{x\in\phi}$ are iid log-normal with $\E \hat{K}_x = P_0$ and $\Var(\log(\hat{K}_x)) =\sigma^2$. After rescaling of $\Pi$ by $(\E \hat{K}_x^\delta)^{1/\delta} = P_0 \exp{(-\sigma^2(1-\delta)/2)}$, we retrieve the iid shadowing model in \cite{Blasz15Poisson}. Now it suffices to show that $\hat{\phi}$ satisfies the homogeneity condition (\ref{eq:homo}). 

For the PPP, 
\begin{align}
\nonumber    \E \hat{\Phi}( b(o,t)) &= \E \sum_{i\geq 0} \mathbbm{1} (\|x_i\|/r(x_i) < t)\\
  \nonumber  & = \sum_{i\geq 0} \P (r(x_i)/\|x_i\| > 1/t)\\
 \nonumber   & \peq{a} t^2.
\end{align}
Step (a) follows from the ccdf of $r(x_i)/\|x_i\|$ given in  Lemma \ref{lemma: |x_0|/r(x)} and Lemma \ref{lemma: r(x_i)/|x_i|}. $\mathbbm{1}(\cdot)$ is the indicator function. Hence we have $\E \hat{\Phi}( b(o,t))/\pi t^2= 1/\pi.$ By the ergodicity of the PVT \cite{Heinrich1994}, $\lim_{t\to\infty}\hat{\phi}( b(o,t))/\pi t^2 = 1/\pi$.
\end{proof}

%% file: Conclusions.tex
\section{Conclusions}
 This paper provides new results on the directional radii of the typical and the zero cell in the Poisson Voronoi tessellations, which characterize the cell shape and unveil the cell asymmetry.  
Based on the directional radii, a joint spatial-propagation model for coverage-oriented cellular networks is studied. In contrast to virtually all prior models, the JSP model ascribes the Poisson deployment of base stations to an intelligent design by the operators, rather than to pure randomness as it would result from a blind placement, ignorant of propagation conditions. As a result, the JSP model with the seemingly pessimistic Poisson deployment performs as well as the standard triangular lattice model. For instance, with $\alpha=4$, there is a 3.4 dB SIR gap between the standard Poisson model and the standard triangular lattice model. Such a gap is eliminated with the JSP model when $\sigma=0$.
This work also highlights the effect of the variance of the large-scale path loss along the cell edge. In the limiting case of $\sigma\to\infty$, the path loss point process for the JSP-PPP converges to that of a PPP. 

{For future work, the effects of shadowing correlation beyond that derived from the cell radius correlation can be analyzed. For instance, the variance of shadowing is usually correlated with distance, and/or the shadowing coefficients from nearby BSs are correlated even for deterministic BS locations. Another interesting direction is the modeling and analysis of  capacity-oriented networks, where  one may ascribe the Poisson deployment to the local user density. In this case, the typical user has a higher chance of being in close proximity to its serving BS.}

%% file: appendix.tex
\appendix
\subsection{Proof of Corollary 1}

Letting $\varphi=0$, the joint distribution of $D_0$, $R_0(0)$ is 
\begin{equation}
    f_{R_0(0),D_0}(x,y) = (2\lambda\pi)^2xy\exp{(-\lambda \pi y^2)},\quad y\geq x \geq 0.
\end{equation}
So the pdf of $R_0(0)$ is
\begin{equation}
  f_{D_o}(y) = \int_{0}^{y}f_{R_0(0),D_0}(x,y)\mathrm{d}x
   = 2(\lambda\pi)^2y^3\exp{(-\lambda \pi y^2)}.
\end{equation}

The ccdf of $R_0(0)-D_0$ given $D_0$ can be written as
\begin{align}
    \P(R_0(0) - D_0 >y\mid D_0 =x) &= \P(R_0(0)> x+y \mid D_0 = x)\nonumber \\
    &= \exp(-\lambda\pi(y^2+2xy)),
\end{align}and
\begin{equation}
    f_{R_0(0) - D_0\mid D_0}(y\mid x)= 2\lambda\pi(x+y)\exp(-\lambda\pi(y^2+2xy)).
\end{equation}
The joint distribution of $R_0(0)$, $R_0(0)-D_0$ is
\begin{equation}
     f_{R_0(0),R_0(0)-D_0}(x,y) = (2\lambda\pi)^2x(x+y)\exp{(-\lambda \pi (x+y)^2)},
\end{equation}
which gives the pdf of $R_0(0)-D_0$ as
\begin{align}
  f_{R_0(0)-D_0}(y) &=(2\lambda\pi)^2 \int_{0}^{\infty}x(x+y)\exp{(-\lambda \pi (x+y)^2)}\mathrm{d}x \nonumber\\
  &= \sqrt{\lambda}\pi \erfc{(y\sqrt{\lambda\pi})}.
\end{align}

For $\varphi = \pi$, we obtain $S(\pi,x,y) = 0$,  $\frac{\partial S(\pi,x,y)}{\partial y}= 0$. $f_{D_0, R_0(\pi)}(x,y) =2\lambda \pi x 2\lambda\pi y\exp(-\lambda\pi (x^2+y^2))=f_{D_0}(x)f_{R_0(\pi)}(y) $.  Thus, $D_0$ and $ R_0(\pi)$ are iid.

\subsection{$R(\varphi)$ in One-Dimensional PPPs}
Let $\Phi$ be a one-dimensional PPP with intensity $\lambda$. Let $X_1, X_{-1}$ be the distances from the origin (the typical point) to the first right and first left point. Let $R_1 = \min\{X_{1}/2, X_{-1}/2\}$ and $R_2 = \max\{X_{1}/2, X_{-1}/2\}$. $R_1$, $R_2$ has the joint pdf
\[f_{{R}_{1}, {R}_{2}}\left(r_{1}, r_{2}\right)=8 \lambda^{2} \exp \left(-2 \lambda\left(r_{1}+r_{2}\right)\right), \quad 0 \leq r_{1} \leq r_{2}.\]
Now,\allowdisplaybreaks
\begin{align}\nonumber
   & \P (R(\pi)\leq r) \\
    &= \E[\P (R(\pi) \leq r\mid {R}_{1}=r_1,{R}_{2}=r_2))]\\\nonumber
& = \E \Big[\frac{r_2}{r_1+r_2} \mathbbm{1}(r_1\leq r\leq r_2) + \mathbbm{1}(r\geq r_2)\Big]
\\\nonumber
&= \int_{r}^{\infty}\int_{0}^{r}\frac{r_2}{r_1+r_2}8\exp(-2(r_1+r_2)) \dd r_1 \dd r_2 + \P (r_2\leq r)
\\\nonumber
& = 1-\exp{(-2\lambda r)}+2\lambda r\exp(-2\lambda r)-4\lambda^2r^2 \Eii (2\lambda r),
\end{align}
where $\Eii(x) = \int_{x}^{\infty}\frac{\exp(-t)}{t} \dd t$ is the exponential integral function. We have $\E R(\pi) = \E D = 1/3$, and $\E R(0) = 2/3$.